%% file: main.tex
\documentclass[a4paper,UKenglish,cleveref, autoref]{lipics-v2019}
\pdfoutput=1

\title{Envy-freeness and Relaxed Stability under lower quotas} 

\author{Prem Krishnaa\footnote{Part of this work was done when the author was a Dual Degree student at IIT Madras.}}{Cohesity Storage Solutions India Pvt. Ltd, India}{premkrishnaa.jaganmohan@cohesity.com}{}{}

\author{Girija Limaye}{Indian Institute of Technology Madras, India}{girija@cse.iitm.ac.in}{}{}
\author{Meghana Nasre}{Indian Institute of Technology Madras, India}{meghana@cse.iitm.ac.in}{}{}
\author{Prajakta Nimbhorkar}{Chennai Mathematical Institute, India and UMI ReLaX}{prajakta@cmi.ac.in}{}{}

\authorrunning{P. Krishnaa, G. Limaye, M. Nasre, P. Nimbhorkar}

\Copyright{Prem Krishnaa, Girija Limaye, Meghana Nasre, and Prajakta Nimbhorkar}



\category{}

\relatedversion{}

\supplement{}

\nolinenumbers 

\hideLIPIcs  

\usepackage[ruled, linesnumbered]{algorithm2e} 
\usepackage{tikz}
\usepackage{subcaption}
\usepackage{tikz}
\usepackage{enumitem}
\usepackage{thm-restate}

\usepackage{wrapfig}
\newcommand{\MINUR}{\mbox{{\sf MIN-UR-EFM}}}
\newcommand{\IS}{\mbox{{\sf IND}\text{-}{{\sf SET}}}}
\newcommand{\NP}{\mbox{{\sf NP}}}
\newcommand{\FPT}{\mbox{{\sf FPT}}}

\newcommand{\Poly}{\mbox{{\sf P}}}
\newcommand{\HR}{\mbox{{\sf HR}}}

\newcommand{\HRLQ}{\mbox{{\sf HRLQ}}}
\newcommand{\LQ}{\mbox{lower-quota}}

\newcommand{\TR}{\mbox{{\sf 01-HRLQ-2R}}}
\newcommand{\CL}{\mbox{{\sf CL}}}
\newcommand{\MAXEFM}{\mbox{{\sf MAXEFM}}}
\newcommand{\MAXRSM}{\mbox{{\sf MAXRSM}}}
\newcommand{\MVC}{\mbox{\sf MVC}}

\newcommand{\ESDA}{\mbox{{\sf ESDA}}}

\newcommand{\CSM}{\mbox{{\sf CSM}}}
\newcommand{\EFHRLQ}{\mbox{{\sf{EF}\text{-}\sf{HR}\text{-}\sf{LQ}}}}
\newcommand{\HH}{\mathcal{H}}
\newcommand{\RR}{\mathcal{R}}
\newcommand{\EE}{\mathcal{E}}

\newenvironment{appendix-lemma}[1]{\vspace{0.1in}\noindent{\bf Lemma~#1~} \em }{\vspace{0.1in}}
\newenvironment{appendix-theorem}[1]{\vspace{0.1in}\noindent{\bf Theorem~#1~} \em }{\vspace{0.1in}}

\begin{document}

\maketitle
\begin{abstract}
We consider the problem of matchings under two-sided preferences in the presence of maximum as well as minimum quota requirements for the
agents. This setting, studied as the Hospital Residents with Lower Quotas ($\HRLQ$) in literature, models important
real world problems like assigning medical interns (residents) to hospitals, and teaching assistants to instructors where a minimum guarantee is essential. 
When there are no minimum quotas, {\em stability} is the de-facto
notion of optimality. 
However, in the presence of minimum quotas, ensuring stability and simultaneously satisfying lower quotas is not an attainable goal in many instances. 

To address this, a relaxation of stability known as {\em envy-freeness}, is proposed in literature.  
In our work, we thoroughly investigate envy-freeness from a computational view point.
Our results show that computing envy-free matchings that match maximum number of agents is computationally hard and also hard to approximate up to a constant factor. Additionally, it is known that envy-free matchings satisfying lower-quotas may not exist.
To circumvent these drawbacks, we propose a new notion called {\em relaxed stability}. We show that relaxed stable matchings are guaranteed to exist
even in the presence of lower-quotas. Despite the computational intractability of finding a largest matching that is feasible and relaxed stable, we
give efficient algorithms that compute a constant factor approximation to this matching in terms of size.
\end{abstract}

\input{1-intro}

\input{2-maxefm}

\input{5-param}
\input{3-maxrsm}

\input{12-concl}

\bibliography{refs}

\appendix
\input{14-appendix}

\end{document}

%% file: 1-intro.tex
\section{Introduction}\label{sec:intro}
Matching problems with two-sided preferences have been extensively investigated  for matching markets  where agents (hospitals/residents or colleges/students)
have upper quotas that can not be exceeded. Stability~\cite{GS62} is a widely accepted notion of optimality in this scenario. An allocation is said to be {\em stable} 
if no pair of agents has an incentive to deviate from it. However, the case when the agents
have maximum as well as minimum quotas poses new challenges and there is still a want of satisfactory mechanisms that take minimum quotas into
account. 
Lower quotas are important from a practical perspective, since it is natural for a hospital to require a minimum number of residents
to run the hospital smoothly. The setting also models important applications like course-allocation, and assigning teaching assistants
(TAs) in academic institutions where a minimum guarantee is essential.

Ensuring stability and at the same time satisfying lower quotas is not an attainable goal in many instances. On one hand,  disregarding
preferences in the interest of satisfying the lower quotas gives rise to social unfairness (for instance agents envying each other); on the other hand, too much emphasis on fairness can lead  to
wastefulness \cite{FITUY15}. Therefore, it is necessary to strike a balance between
these three mutually conflicting goals -- optimality with respect to preferences, feasibility for minimum quotas and minimizing wastefulness.
The main contribution of this paper is to propose a mechanism  to achieve this balance.

Envy-freeness~\cite{FITUY15,GIKKYY16,KK15,KK17,EHYY14} is a widely accepted notion for achieving fairness from a social perspective. 
Unfortunately, the two goals viz. envy-freeness and feasibility may not be simultaneously achievable; whether feasible envy-free matchings exist can be answered efficiently by the characterization of Yokoi~\cite{Yokoi20}.
Fragiadakis~et~al.~\cite{FITUY15} explore strategy-proof aspects of envy-freeness and the trade-off between envy-freeness and wastefulness; however their results are for a restricted setting of agent preferences.
In our work, we thoroughly investigate envy-freeness from a computational view point both in the general and restricted settings.
Our results show that computing envy-free matchings that match maximum number of agents is computationally hard and even such matchings can be wasteful.
To circumvent these drawbacks, we propose a new notion called {\em relaxed stability}. We show that relaxed stable matchings are guaranteed to exist
even in the presence of lower-quotas. We additionally show that a relaxed stable matching that is at least the size of the stable matching in the instance (disregarding lower quotas) exists and can be efficiently computed. On the other hand, if we insist for the largest size relaxed stable matching, computing such a matching turns out to be computationally intractable.

We state the problem formally in terms of assigning a set of medical interns (residents) to a set of hospitals
where preferences are expressed by both the sets, and hospitals have upper quotas 
and lower quotas associated with them. This is called the $\HRLQ$ setting in literature. 
The input in the $\HRLQ$ setting is a bipartite graph $G = (\RR \cup \HH, E)$ where  $\RR$ denotes the set of
residents, $\HH$ denotes the set of hospitals, and an edge $(r, h) \in E$ denotes that $r$ and $h$ are mutually
acceptable. 
A hospital $h \in \HH$ has an 
upper-quota $q^+(h)$ which denotes the maximum
number of residents that can be assigned to $h$. Additionally, $h$ has a lower-quota $q^-(h)$ which denotes the minimum number of 
residents that must be assigned to $h$. 
Finally, every vertex (resident and hospital) in $G$ ranks its neighbors in a strict order, referred to as the {\em preference list} of the vertex. 
If a vertex $a$ prefers its neighbor $b_1$ over $b_2$, we denote it by $b_1 >_a b_2$.

A matching $M \subseteq E$ in $G$ is an
assignment of residents to hospitals such that 
 each resident is 
matched to at most one hospital, and every hospital $h$ is matched to 
at most $q^+(h)$-many residents.
For a matching $M$, let $M(r)$ denote the hospital that $r$ is matched to, and $M(h)$ denote the set of residents
matched to $h$ in $M$.
If resident $r$ is unmatched in matching $M$ we let $M(r) = \bot$ and $\bot$ is considered as the least preferred choice by any resident. 
We say that a hospital $h$ is {\em under-subscribed} in $M$ if $|M(h)| < q^+(h)$, $h$ is {\em fully-subscribed} in $M$ if $|M(h)| = q^+(h)$ and $h$ is {\em deficient}
in $M$ if $|M(h)| < q^-(h)$. A matching is {\em feasible} for an $\HRLQ$ instance if no hospital is deficient in 
$M$. 
The goal for the $\HRLQ$ problem is to match residents to hospitals {\em optimally} with respect to the preference lists such that the matching is feasible.
The $\HRLQ$ problem is a generalization of the well-studied $\HR$ problem (introduced by Gale and Shapley~\cite{GS62}) where there are no lower quotas. 
In the $\HR$ problem, {\em stability} is a de-facto notion of optimality and is defined by the absence of {\em blocking pairs}.
\begin{definition}[Stable matchings]\label{def:stable}
A pair $(r, h) \in E \backslash M$ is a blocking pair w.r.t. the matching $M$ if $h >_r M(r)$ and $h$ is either under-subscribed in $M$ or there exists at least one resident $r' \in M(h)$ such that $r >_h r'$. A matching $M$ is {\em stable} if there is no blocking pair w.r.t. $M$.
\end{definition}

\noindent {\bf Existence of Stable Feasible Matchings:} Given an $\HRLQ$ instance, it is natural to ask ``does the instance admit a stable feasible matching?"
A stable matching always exists in an $\HR$ instance and can be computed in linear time~\cite{GS62}. In contrast, there exist simple $\HRLQ$ instances 
that may not admit any feasible and stable matching, even when a feasible matching exists.
\begin{wrapfigure}{r}{0.5\textwidth}
\begin{center}
    \begin{minipage}{0.2\textwidth}
    \begin{align*}
        r_1 &: h_1, h_2\\
        r_2 &: h_1
    \end{align*}
    \end{minipage}%
    \begin{minipage}{0.25\textwidth}
    \begin{align*}
        \text{[0,1]}\text{ } {h_1} &: r_1, r_2\\
        \text{[1,1]}\text{ } {h_2} &: r_1
    \end{align*}
    \end{minipage}%
\caption{An $\HRLQ$ instance with no feasible and stable matching.}
\label{fig:no_stb_feas}
\end{center}
\end{wrapfigure}
\raggedbottom
For example, Fig.~\ref{fig:no_stb_feas} shows an $\HRLQ$ instance with 
$\RR = \{r_1, r_2\}$ and $\HH = \{h_1, h_2\}$.
	where both hospitals have a unit upper-quota and $h_2$ has a unit lower-quota.
	We denote the lower-quota and upper-quota of hospital $h$ using $[q^-(h), q^+(h)]$ before hospital $h$.
The matching $M_s = \{(r_1, h_1)\}$ is stable but not feasible since $h_2$ is deficient
in $M_s$. The matchings $M_1 = \{ (r_1, h_2) \}$ and $M_2 = \{(r_1, h_2), (r_2, h_1)\}$ 
are both feasible but not stable since $(r_1, h_1)$ blocks both of them. 
The existence question of a stable, feasible matching for $\HRLQ$ can be answered by simply ignoring lower quotas
and computing a stable matching in the resulting $\HR$ instance. From the well-known Rural Hospitals Theorem~\cite{Roth86}, the {\em number} of residents matched to a hospital is invariant
across all stable matchings of the instance. Hence, for any $\HRLQ$ instance, either all the stable matchings are feasible or all are infeasible, and they have the same size.

Imposing lower quotas ensures that infeasible matchings
are no longer acceptable, however, the presence of lower quotas poses new challenges as discussed above.
In light of the fact that stable and feasible matchings may not exist, relaxations of stability, like popularity and envy-freeness have been proposed in the literature \cite{NN17,MNNR18,Yokoi20}. Envy-freeness is defined by the absence of {\em envy-pairs}.
\begin{definition}[Envy-free matchings]\label{def:ef}
Given a matching $M$, 
	a resident $r$ has a {\em justified envy} (here onwards called envy) towards a matched resident $r'$, where $M(r') = h$ and $(r,h) \in E$
if $h >_r M(r)$ and $r >_h r'$. The pair $(r, r')$ is
an envy-pair w.r.t. $M$.  A matching $M$ is envy-free
if there is no envy-pair w.r.t. $M$.
\end{definition}
Note that an envy-pair implies a blocking pair but the converse is not true and hence envy-freeness is a relaxation of stability.
In the example in Fig.~\ref{fig:no_stb_feas}, the matching $M_1$ is envy-free 
and feasible, although not stable. Thus,
envy-free matchings provide an alternative to stability in such instances. 
Envy-freeness is motivated by fairness from a social perspective. Importance of envy-free matchings has been recognized in the context of 
constrained matchings \cite{FITUY15,GIKKYY16,KK15,KK17,EHYY14}, and their
structural properties have been investigated in \cite{WR18}. 
\begin{wrapfigure}{r}{0.5\textwidth}
    \begin{minipage}{0.2\textwidth}%
    \begin{align*}
	    \forall i \in [n],\ r_i &: h_1, h_2
    \end{align*}
    \end{minipage}
    \begin{minipage}{0.28\textwidth}
    \begin{align*}
        \text{$[0,n]$}\text{ } {h_1} &: r_1, \cdots, r_n\\
        \text{$[1,1]$}\text{ } {h_2} &: r_1, \cdots, r_n
    \end{align*}
    \end{minipage}%
\caption{An $\HRLQ$ instance with two envy-free matchings of different sizes.}
\label{fig:min_max_efms}
\end{wrapfigure}

\noindent {\bf Size of envy-free matchings:} 
In terms of size, there is a sharp contrast between stable matchings in the $\HR$ setting and envy-free matchings in the $\HRLQ$ setting.
While all the stable matchings in an $\HR$ instance have the same size,
the envy-free matchings in an $\HRLQ$ instance may have significantly different sizes.
Consider the example in Fig.~\ref{fig:min_max_efms} with $n$ residents $\RR = \{r_1, r_2, \cdots, r_n\}$  and 
two hospitals $\HH = \{h_1, h_2\}$. 
The hospital $h_2$ has a unit upper-quota and a unit lower-quota. The instance admits an envy-free matching $N_1 = \{(r_1, h_2)\}$
of size one and another envy-free matching $N_n = \{(r_1, h_1), (r_2, h_1), \ldots,$ $(r_{n-1}, h_1), (r_n, h_2)\}$ of size $n$.

\noindent{\bf Shortcomings of envy-free matchings:} It is interesting to note that a feasible, envy-free matching 
itself may not exist -- for instance, if the example in 
Fig.~\ref{fig:no_stb_feas} is modified such that both $h_1, h_2$ have a unit lower-quota, then $M_2$
is the unique feasible matching. However, $M_2$ is not envy-free since $(r_1, r_2)$ is an envy-pair
w.r.t. $M_2$.
If a stable matching is not feasible in an $\HRLQ$ instance, {\em wastefulness} may be inevitable for attaining feasibility.
A matching is wasteful if there exists a resident who prefers a hospital to her current assignment and that hospital has a vacant position~\cite{FITUY15}.
Envy-free matchings can be significantly wasteful. For instance, in the example in Fig.~\ref{fig:min_max_efms}, the matching $N_1$ is wasteful.
Therefore, it would be ideal to have a notion of optimality which is guaranteed to exist, is efficiently computable and avoids wastefulness.

\noindent{\bf Quest for a better optimality criterion: } We propose a new notion of {\em relaxed stability} which always exists for any $\HRLQ$ instance.
We observe that 
in the presence of lower quotas, 
there can be at most $q^-(h)$-many residents that are forced to be matched to $h$, even though they have higher preferred under-subscribed 
hospitals in their list.
Our relaxation allows these forced residents to participate in blocking-pairs,\footnote{Our initial idea was to allow them to participate in envy-pairs. We thank anonymous reviewer for suggesting this modification which is stricter than our earlier notion.}
\ however, 
the matching is still stable when restricted to the remaining residents. We now make this formal below.

\begin{definition}[Relaxed stable matchings]
	A matching $M$ is {\em relaxed stable} if, for every hospital $h$, at most $q^-(h)$ residents from $M(h)$ participate in blocking pairs 
	and no unmatched resident participates in a blocking pair. 
\end{definition}

\begin{wrapfigure}{r}{0.5\textwidth}
    \begin{minipage}{0.2\textwidth}
    \begin{align*}
	    r_1 &: h_1, h_3\\
	    r_2 &: h_2, h_3\\
	    r_3 &: h_2\\
    \end{align*}
    \end{minipage}\hfill
    \begin{minipage}{0.25\textwidth}
    \begin{align*}
        \text{$[0,1]$}\text{ } {h_1} &: r_1\\
        \text{$[0,1]$}\text{ } {h_2} &: r_2, r_3\\
        \text{$[1,1]$}\text{ } {h_3} &: r_1, r_2\\
    \end{align*}
    \end{minipage}%
\caption{An $\HRLQ$ instance with two relaxed stable matchings of different sizes, one larger than stable matching}
\label{fig:min_max_rsm}
\end{wrapfigure}

\raggedbottom
We note that in the example in Fig.~\ref{fig:no_stb_feas}, the matching $M_2$ (which was not envy-free) is feasible, relaxed stable and non-wasteful. 
In fact, we show that every instance of the $\HRLQ$ problem admits a feasible relaxed stable matching -- thus 
addressing the issue of guaranteed existence. 
In terms of computation, a relaxed stable matching can be efficiently computed; however if we insist on the maximum size relaxed stable matching, we show that this problem
	is computationally hard.

	In order to ensure guaranteed existence of relaxed stable matchings, we need to allow upto $q^-(h)$ many residents per hospital to participate in blocking pairs. We remark that Hamada et al.~\cite{HIM16} studied a similar notion of computing matchings with minimum blocking pairs. Such matchings ($\sf{MINBP}$) are guaranteed to exist, however, computing them is $\NP$-hard even under severe restrictions on the preference lists. Contrast this with relaxed stability which is guaranteed to exist and a relaxed stable matching at least as large as a stable matching (obtained by disregarding lower-quotas) is efficiently computable.
In fact, a relaxed stable matching may be even larger than the size of the stable matching in the instance, as seen in the example in Fig.~\ref{fig:min_max_rsm}.
In this instance, the stable matching $M_s = \{(r_1, h_1), (r_2, h_2)\}$ of size two is infeasible. Matchings $M_1' = \{(r_1, h_3), (r_2, h_2)\}$ and $M_2' = \{(r_1,h_1), (r_2, h_3), (r_3, h_2)\}$
both are relaxed stable and feasible and $M_2'$ is larger than $M_s$. 
	This is in contrast to maximum size envy-free matching which (as we will see in section~\ref{sec:hardness}) cannot be larger than the size of a  stable matching. 

\noindent{\bf Our contributions: } 
	In this paper, we study the computational complexity, approximability, parameterized complexity and the hardness of approximation for the notions of 
envy-freeness and relaxed stability. Throughout, we assume that our input $\HRLQ$ instance admits some feasible matching and our algorithms always aim
to output a feasible matching that is optimal.
We consider 
the problem of computing a maximum-size, feasible, optimal matching in an $\HRLQ$ instance when one exists.
When the optimality criterion is envy-freeness, we denote this as the $\MAXEFM$ problem, and the equivalent problem of computing an envy-free matching that has the minimum number of unmatched residents 
as the $\MINUR$ problem. For exact solutions, the two problems are equivalent. 
When the optimality criterion is relaxed stability, we denote this as the $\MAXRSM$ problem.

\noindent{\bf{Results on envy-freeness: }} 
We show that the $\MAXEFM$ problem is NP-hard, and in fact, is hard to approximate below a constant factor.
\begin{theorem}
	The following hold:
	\label{thm:strong-inapprox}
	\begin{enumerate}[label={(\Roman*)}]
	\item The {$\MAXEFM$} (equivalently $\MINUR$) problem is $\NP$-hard.
	Moreover, the $\MINUR$ problem has no $\alpha$-approximation algorithm for any factor $\alpha > 0$ 
	unless $\Poly=\NP$. Above hardness results hold when
	\label{part1}
			\begin{enumerate}[label={(\alph*)}]
		\item every resident has a preference list of length at most two (upper-quotas of hospitals can be arbitrary).\label{parta}
		\item every hospital has lower-quota and upper-quota at most one (resident preference lists can be longer than two).\label{partb}
	\end{enumerate}
	\item The $\MAXEFM$ problem can not be approximated within a factor of $\frac{21}{19}$ unless $\Poly=\NP$ {even when every hospital has a quota of at most one}. 
			\label{part2}
	\end{enumerate}
\end{theorem}

\noindent In light of the above negative result, we consider $\MAXEFM$ problem on restriction on $\HRLQ$ instance, called the $\CL$-restriction~\cite{HIM16}. The restriction requires that every hospital with positive lower-quota must rank every resident, and consequently, every resident
ranks every hospital with a positive lower quota.
Note that in Fig.~\ref{fig:no_stb_feas}, the infeasibility of hospital $h_2$ could be resolved if $h_2$ and $r_2$ were mutually acceptable. In that case, the stable matching $\{(r_1,h_1), (r_2,h_2)\}$ is feasible
and hence is a maximum size envy-free matching.
The $\CL$-restriction has been considered by Hamada~et~al~\cite{HIM16} where the goal is to output a matching with minimum number of blocking pairs ($\sf{MINBP}$) or blocking residents ($\sf{MINBR}$). Hamada~et~al.~\cite{HIM16} 
proved that even under the $\CL$-restriction, computing the $\sf {MINBP}$ and $\sf{MINBR}$ problems are $\NP$-hard. In contrast under the $\CL$-restriction, $\MAXEFM$ (equivalently $\MINUR$) is tractable.
\begin{theorem}\label{thm:CL}
	There is a simple linear-time algorithm for the $\MAXEFM$ (equivalently $\MINUR$) problem for {\sf CL}-restricted $\HRLQ$ instances.
\end{theorem}
In practice it is common to have preference lists which are incomplete and in many cases the preference lists of residents may also be constant size.
Krishnapriya et al.~\cite{MNNR18} present an algorithm that efficiently computes a {\em maximal} envy-free matching that extends a given feasible matching. 
Matching $M$ is a \emph{maximal envy-free matching} if addition of any edge to $M$ violates either the upper-quota or envy-freeness.
However, prior to this work, no approximation guarantee of a maximal envy-free matching was known. We prove following guarantee on the size of a maximal envy-free matching. Let $\ell_1$ be the length of the longest preference list of a resident and $\ell_2$ be the length of the longest preference list of a hospital.
\begin{theorem}\label{thm:app_maximal}
	A maximal envy-free matching is 
	\begin{enumerate}[label={(\Roman*)}]
		\item an $\ell_1$-approximation of $\MAXEFM$ when hospital quotas are at most $1$\label{p1}
		\item an $(\ell_1\cdot\ell_2)$-approximation of $\MAXEFM$ when quotas are unrestricted.\label{p2}
	\end{enumerate}
\end{theorem}

\noindent 
Besides the above results, we investigate the parameterized complexity of the problem.
When the stable matching is not feasible, there is at least one lower-quota hospital that is deficient.
\emph{Deficiency}~\cite{HIM16} of an $\HRLQ$ instance with respect to a stable matching $M$ is defined as follows.
\begin{definition}\label{def:deficiency}
	Let $G=(\RR\cup \HH, E)$ be an \HRLQ\ instance and $M$ be a stable matching in $G$. Then deficiency($M$) = $\sum\limits_{h \in \HH}{\max\{0, q^-(h) - |M(h)|\}}$.
\end{definition}
We show that $\MAXEFM$ and $\MINUR$ are W[1]-hard when deficiency is the parameter. We also show a polynomial size kernel and present FPT algorithms for the $\MAXEFM$ problem. The respective parameters are defined in section~\ref{sec:pc_results}.
\begin{theorem}\label{thm:pc_results}
	The following hold:
	\begin{enumerate}[label={(\Roman*)}]
		\item The $\MAXEFM$ and $\MINUR$ are W[1]-hard when deficiency is the parameter. The hardness holds even when residents preference lists are of length at most two or hospital quotas are $0$ or $1$.\label{pc_p1}
		\item The $\MAXEFM$ has a polynomial size kernel.\label{pc_p2}
		\item The $\MAXEFM$ admits $\FPT$ algorithms for several interesting parameters.\label{pc_p3}
	\end{enumerate}
\end{theorem}

\noindent{\bf {Results on relaxed stability: }} 
We prove that the $\MAXRSM$ problem is $\NP$-hard and is also hard to approximate, but has a better approximation behavior than the $\MAXEFM$ problem.
\begin{theorem}\label{thm:inapprox_rsm}
	The $\MAXRSM$ problem is $\NP$-hard and cannot be approximated within a factor of $\frac{21}{19}$ unless $\Poly=\NP$. The result holds even when all quotas are at most one.
\end{theorem}
We complement the above negative result with the following:
\begin{theorem}\label{thm:approx-algo}
	Any feasible $\HRLQ$ instance always admits a relaxed-stable matching. Moreover, there is a polynomial-time algorithm that outputs a $\frac{3}{2}$-approximation to
the maximum size relaxed stable matching.

\end{theorem}

\noindent We summarize our results in Table~\ref{tab:summary_results}.
\begin{table}[!ht]
	\begin{tabular}{|p{0.15\textwidth}|p{0.22\textwidth}|p{0.23\textwidth}|p{0.3\textwidth}|}
\hline
	{\bf Problem} & {\bf Hardness and Inapproximability} & {\bf Approximation and parameterized results} & {\bf Restricted settings}\\
\hline
		\rule{0pt}{2.3ex} $\MAXEFM$ & $\frac{21}{19}$-inapproximability , W[1]-hard w.r.t. deficiency &  $(\ell_1\cdot\ell_2)$-approximation, Polynomial kernel, FPT for several parameters & P-time for CL-restriction,\\
& & & $\ell_1$-approximation for $0/1$ quotas\\
\hline
	$\MINUR$ & $\alpha$-inapproximability for any $\alpha>0$, W[1]-hard w.r.t. deficiency& --& P-time for CL-restriction\\
\hline
	\rule{0pt}{2.3ex} \rule[-1.2ex]{0pt}{0pt} $\MAXRSM$ & $\frac{21}{19}$-inapproximability & $\frac{3}{2}$-approximation & --  \\
\hline
\end{tabular}
	\caption {Summary of our results}
	\label{tab:summary_results}
\end{table}
\raggedbottom

\noindent{\bf Related work:}
Various notions of optimality in the $\HRLQ$ setting have been studied in \cite{FITUY15,HIM16,Yokoi20,NN17,MNNR18}. 
Hamada~et~al.~\cite{HIM16} consider computing feasible matchings with minimum number of blocking pairs or blocking residents. However
both these objectives are $\NP$-hard even under severe restrictions. 
A trade-off between envy-freeness and non-wastefulness is considered in \cite{FITUY15}.
Another notion of optimality, namely {\em popularity} in the $\HRLQ$ problem has been considered in \cite{NN17}.
Popularity can be regarded as {\em overall stability}. It was shown in \cite{NN17} that
 a matching which is popular amongst feasible matchings always
exists. On the flip-side, a popular matching is not guaranteed to be either envy-free or even relaxed stable.
Strategy-proof mechanisms for the lower quota setting are presented in~\cite{FITUY15}.
In a different setting with lower quotas, in which hospitals either fulfill required lower quotas or are closed 
is studied in~\cite{BFIM10}. 
Lower quotas are also studied by \cite{Huang10} and \cite{FK16} in the context of classified stable matchings ($\CSM$).
Parameterized complexity for the problem of computing maximum size stable matching with ties and incomplete lists (without lower quotas) is studied in~\cite{AGRSZ18}.

\vspace{0.2cm}
\noindent{\em Organization of the paper: }
In section~\ref{sec:maxefm_hardness}, we present our results related to $\NP$-hardness and hardness of approximation of $\MAXEFM$ and $\MINUR$ problems.
In section~\ref{sec:maxefm_positive}, we present our algorithmic results for $\MAXEFM$ and $\MINUR$ 
and our approximation results for the $\MAXEFM$ problem.
In section~\ref{sec:pc_results}, we present our parameterized complexity results for the $\MAXEFM$ problem.
In section~\ref{sec:maxrsm}, we present hardness of approximation of $\MAXRSM$ followed by an approximation algorithm. 
In section~\ref{sec:conl}, we conclude and discuss open problems.

%% file: 2-maxefm.tex
\section{Envy-freeness: Hardness and Inapproximability}\label{sec:maxefm_hardness}
In this section we prove $\NP$-hardness of the $\MAXEFM$ and $\MINUR$ problems for arbitrary preference lists. 
We also prove that $\MINUR$ problem cannot be approximated for any $\alpha >0$ and that
$\MAXEFM$ problem cannot be approximated for a factor within $\frac{21}{19}$ unless $\Poly = \NP$.
\input{2-minUR}
\input{3-inapprox}
\section{Envy-freeness: Algorithmic  results}\label{sec:maxefm_positive}
Given the $\NP$-hardness of $\MAXEFM$ and $\MINUR$ problems, we turn our attention to special cases
of $\HRLQ$ instances which are tractable. One such restriction is the $\CL$-restriction.
Next, we prove an approximation guarantee of any maximal envy-free matching. 
\input{4-cl}
\subsection{Approximation guarantee of a maximal envy-free matching}
As mentioned earlier, Krishnapriya et al.~\cite{MNNR18} present an algorithm to compute a maximal envy-free matching that extends a given envy-free matching.
However, their results are empirical and no theoretical guarantees are known about the size of a maximal envy-free matching.
In this section we present approximation guarantee of a maximal envy-free matching. 
Below we prove the first part of the Theorem~\ref{thm:app_maximal} for the restricted instance where hospital quotas are at most $1$.
\begin{proof}[Proof of Theorem~\ref{thm:app_maximal}\ref{p1}]
Let $M$ be a maximal envy-free matching and $OPT$ be a $\MAXEFM$. Let $R_{OPT}$ and $R_M$ denote the set of residents matched in $OPT$ and $M$ respectively.
Let $X_1$ be the set of residents matched in both $M$ and $OPT$. Let $X_2$ be the set of residents matched in $OPT$ but not matched in $M$.
        Thus, $|R_{OPT}| = |X_1| + |X_2|$.
        Since $X_1 = R_{OPT} \cap R_M \subseteq R_M$, so $|X_1| \leq |R_M|$.
        Our goal is to show that 
        $|X_2| \le |R_M| \cdot (\ell_1 - 1)$. Once we establish that,
	it is immediate that a maximal envy-free matching is an $\ell_1$-approximation.

We show that for every resident $r \in X_2$ we can associate
a unique hospital $h_r$ such that $h_r$ is unmatched in $M$
and there exists a resident $r'$ in the neighbourhood of $h_r$ such that $r'$ is matched in $M$.
Denote the set of such hospitals as $Y_2$. Note that due to the uniqueness assumption $|X_2| = |Y_2|$.
Since each resident has a preference list of length at most $\ell_1$, any $r'$ who is matched in $M$ can have at most $\ell_1-1$ neighbouring hospitals which are unmatched in $M$.
Thus $|X_2| = |Y_2| \le |R_M| \cdot (\ell_1 - 1)$ which establishes the
approximation guarantee.
To finish the proof we show a unique hospital $h_r$ with desired properties that can be associated with each $r \in X_2$.
Let $r \in X_2$ such that $h = OPT(r)$.
We have following two exhaustive cases.
    
{\em Case 1:} If $h$ is unmatched in $M$, then due to maximality of $M$, there must exist a resident $r'$ matched in $M$ such that adding $(r, h)$ causes envy to $r'$.
        Thus, $h$ has a neighboring resident $r'$ matched in $M$, and we let $h_r = h$.

{\em Case 2:} If $h$ is matched in $M$, then since $M$ and $OPT$ are both envy-free, there must exist a path $\langle r, h, r_1, h_1$ $, \ldots, r_i, h_i \rangle$ such that $(r, h) \in OPT$, for each $k = 1, \ldots, i$, we have $(r_k, h_k) \in OPT$, $(r_1, h) \in M$, for each $k = 2, \ldots, i$, we have $(r_k, h_{k-1}) \in M$ and $h_i$ is unmatched in $M$.
    Thus, $h_i$ has a neighboring resident $r_i$ matched in $M$, and we let $h_r = h_i$.

\noindent {\em Uniqueness guarantee:} For any $r \in X_2$ for which case 1 applies, the associated $h_i$ is unique since hospital quotas are at most $1$. For two distinct $r, r' \in X_2$ such that for both case 2 applies, the paths mentioned above are disjoint since all hospital quotas are at most $1$, which guarantees uniqueness within case 2. The $h_i$ associated in case 2 cannot be associated in case $1$ to $OPT(h_i)$ since $OPT(h_i) = r_i \notin X_2$.
This completes the proof of existence of the unique hospital.
\end{proof}

Now, we prove the second part of the Theorem~\ref{thm:app_maximal} for the unrestricted quotas.
\begin{proof}[Proof of Theorem~\ref{thm:app_maximal}\ref{p2}]
	We will use the $R_M, R_{OPT}, X_1, X_2$ sets as defined earlier in the proof of Theorem~\ref{thm:app_maximal}\ref{p1}. 
It is clear that $|X_1| \leq |R_M|$.
We will show that $|X_2| \leq |R_M| \cdot (\ell_1 \cdot \ell_2 - 1)$. Once we establish that,
it is immediate that a maximal envy-free matching is an $(\ell_1 \cdot \ell_2)$-approximation.
We show that for every resident $r \in X_2$ we can associate
a unique edge $(h_r, r_t)$ such that $h_r$ is under-subscribed in $M$
and there exists a resident $r_t$ in the neighbourhood of $h_r$ such that $r_t$ is matched in $M$.
Denote the set of such hospitals as $Y_2$. Because each hospital can have at most $\ell_2$ edges (and all of them could be matched in $OPT$), thus $|X_2| \leq |Y_2| \cdot \ell_2$.
Since each resident has a preference list of length at most $\ell_1$, any $r'$ who is matched in $M$ can have at most $\ell_1$ neighbouring hospitals which are under-subscribed in $M$.
Resident $r'$ is matched to one of these hospitals, thus, $|X_2| \le |R_M| \cdot (\ell_1 \cdot \ell_2 - 1)$ which establishes the approximation guarantee.
To finish the proof we show a unique edge $(h_r, r_t)$ at a hospital $h_r$ with desired properties that can be associated with each $r \in X_2$.
Let $r \in X_2$ such that $h = OPT(r)$.
We have following two exhaustive cases.

{\em Case 1:} If $h$ is under-subscribed in $M$, then due to maximality of $M$, there must exist a resident $r'$ matched in $M$ such that adding $(r, h)$ causes envy to $r'$.
        Thus, $h$ has a neighboring resident $r'$ matched in $M$, i.e. we let $(h_r, r_t) = (h,r')$.

{\em Case 2:} If $h$ is fully-subscribed in $M$, then since $M$ and $OPT$ are both envy-free, there must exist a path $\langle r, h, r_1, h_1$ $, \ldots, r_i, h_i \rangle$ (hospitals can repeat along this path) such that $(r, h) \in OPT$, for each $k = 1, \ldots, i$, we have $(r_k, h_k) \in OPT$, $(r_1, h) \in M$, for each $k = 2, \ldots, i$, we have $(r_k, h_{k-1}) \in M$ and all hospitals $h, h_1, \ldots, h_{i-1}$ are fully-subscribed in $M$ and $h_i$ is under-subscribed in $M$.
    Thus, $h_i$ has a neighbouring resident $r_i$ matched in $M$, i.e. we let $(h_r, r_t) = (h_i,r_i)$.

\noindent {\em Uniqueness guarantee:} For any $r \in X_2$ for which case 1 applies, the associated $(h, r')$ edge is unique. For $r \in X_2$ such that case 2 applies for $r$, at each hospital $h' \in \{h, h_1, \ldots h_{i-1}\}$, if $h'$ has $k$ $OPT$-edges incident on it, then $h'$ must have at least $k$ $M$-edges incident on it such that all $M$-edge partners are higher preferred than all $OPT$-edge partners. Thus, if $h'$ is shared across paths starting at multiple residents in $X_2$, there must exist a unique $M$-edge that extends the specific path for which leads to a unique $OPT$-edge to the next hospital, otherwise $h'$ is under-subscribed in $M$, a contradiction. This guarantees uniqueness of edge within case 2. The $(h_i, r_i)$ associated in case 2 cannot be associated in case $1$ to $r_i$ since $r_i \notin X_2$.
This completes the proof of existence of the unique edge at a desired hospital and also the proof of the lemma.
\end{proof}

\subsection{Polynomial time algorithm for $\MAXEFM$ for a restricted setting}\label{sec:012R}
We present an algorithm (Algorithm~\ref{algo:maxl-envy-free}) that computes a maximum envy-free matching when the resident lists are of length at most two
and hospitals quotas are at most one. 
At a high-level Algorithm~\ref{algo:maxl-envy-free} works as follows: It starts with any feasible envy-free matching (possibly output of Yokoi's $\EFHRLQ$ Algorithm~\cite{Yokoi20}) 
and computes an \emph{envy-free augmenting path} with respect to the current matching. An augmenting path $P$ with respect to an envy-free matching $M$ is envy-free if $M \oplus P$ is envy-free. To compute such a path, the algorithm deletes a set of edges from the graph.
To define these deletions, we need the following definition from \cite{MNNR18}.
\begin{definition} \cite{MNNR18}
	Let $h$ be any hospital in $G$, a {\em threshold resident} $r'$ for $h$, if one exists,
	is the most preferred resident of $h$ such that $h >_{r'} M(r')$. If no such resident exists, we assume a unique dummy resident $d_h$ at the end of $h$'s preference list to be the threshold resident for hospital $h$.
\end{definition}

\begin{algorithm}[!ht]
	\SetAlgoNoLine
	\SetAlgoNoEnd
	\KwIn {An $\HRLQ$ instance $G = (\RR \cup \HH, E)$}
	\KwOut {A maximal envy-free matching in $G$}
         Let $M$ be any feasible envy-free matching in $G$ (assume that $M$ exists)\;
	\Repeat{true}
	{
     	Compute the threshold resident $r'$ for all $h\in \HH$ w.r.t. $M$\; 
	Let $G' = (\RR \cup \HH , E')$ be an induced sub-graph of $G$, where $E' = E \setminus (E_1 \cup E_2)$,\newline 
	    \text{ }\text{ }\text{ } $E_1$ = \{$(r,h) |$  $(r, h) \in E \setminus M$, $r' >_h r$\},\newline
		\text{ }\text{ }\text{ } $E_2$ = \{$(r,h) |$  $(r, h) \in E \setminus M $, $r >_h r'$ and $M(r) >_r h$\}\label{step:del}\;
	   \If {there exists an augmenting path $P$ w.r.t. $M$ in $G'$}{
		  $M_P = M \oplus P$\;
		  $M=M_P$\;
	   	}
		\Else
		{
		 Exit the loop.
		 }
	
	}
 Return $M$\;
\caption{Algorithm to compute a maximal envy-free matching}
\label{algo:maxl-envy-free}
\end{algorithm}
In every iteration of our algorithm, we compute the threshold resident for every hospital. Observe that unless the threshold
resident $r'$ is matched to  $h$ or to a hospital $h' >_{r'} h$, no resident $r'' <_{h} r'$ can be matched to $h$. Thus,
in our algorithm, we delete the set of edges $E_1$ (see Step~\ref{step:del}) which correspond to lower preferred residents than
the threshold resident for a hospital. If $r >_h r'$ such that $M(r) \neq h$ and $M(r) >_r h$ then we delete edge $(r,h)$ (set of $E_2$ edges).
Note that both these deletions ensure that
after augmentation, a resident never gets demoted (matched to a lower preferred hospital).
It is easy to see that the algorithm has at most $n$ iterations each taking $O(m+n)$ time, thus the overall running time
is $O(mn)$.

As the proof of Theorem~\ref{thm:strong-inapprox}\ref{part1} mentions, the $\NP$-hardness applies to the $\HRLQ$ instances
in which either the residents have preference list of length at most $2$ or hospital quotas are at most one.
The $\HRLQ$ instance in Fig.~\ref{fig:cex_a} has quotas at most one but a resident with a three-length list. The example instance
in Fig.~\ref{fig:cex_b} has lengths of all residents at most two but the quotas are not at most one. 
In both the examples, we have an initial envy-free matching $M$ and a larger envy-free matching $M'$. However, there is no envy-free augmenting path w.r.t. $M$.
In contrast, when we impose the restriction that all quotas are at most one {\em and} every resident's list is of length at most two (denoted as $\TR$ restriction),
we have the desired envy-free augmenting paths. We prove this guarantee below.
\begin{figure}
	\begin{subfigure}[b]{0.45\textwidth}
	\begin{minipage}{0.5\columnwidth}
		\begin{align*}
    			r_1 &: h_1,h_4 \\
			r_2 &: h_2,h_3,h_4 \\
		    	r_3 &: h_3
		\end{align*}
	\end{minipage}%
	\begin{minipage}{0.5\columnwidth}
		\begin{align*}
			[0,1]\ h_1 &: r_1 \\
			[0,1]\ h_2 &: r_2 \\
			[0,1]\ h_3 &: r_2,r_3 \\
			[1,1]\ h_4 &: r_1,r_2
		\end{align*}
	\end{minipage}\\
	\begin{center}
	$M = \{(r_1,h_1), (r_2, h_4)\}$\hfill$M' = \{(r_1,h_4), (r_2,h_2), (r_3,h_3)\}$
	\end{center}
	\subcaption{Quotas at most one, but resident preference lists longer than two.}
	\label{fig:cex_a}
	\end{subfigure}\hfill%
	\begin{subfigure}[b]{0.45\textwidth}
	\begin{minipage}{0.5\columnwidth}
		\begin{align*}
    			r_1 &: h_1, h_2 \\
			r_2 &: h_1\\
		    	r_3 &: h_3, h_2\\
		\end{align*}
	\end{minipage}%
	\begin{minipage}{0.5\columnwidth}
		\begin{align*}
			[0,2]\ h_1 &: r_1, r_2 \\
			[1,1]\ h_2 &: r_1, r_3 \\
			[0,1]\ h_3 &: r_3 \\
		\end{align*}
	\end{minipage}
	\begin{center}
	$M = \{(r_1, h_2), (r_3, h_3)\}$\hfill$M' = \{(r_1, h_1), (r_2, h_1), (r_3, h_2)\}$
	\end{center}
	\subcaption{Resident preference lists at most length two, but quotas more than $1$}
	\label{fig:cex_b}
	\end{subfigure}%
	\caption{$\HRLQ$ instances and envy-free matchings $M$ and $M'$} 
	\label{fig:cex}
\end{figure}

\begin{lemma}\label{lem:structural}
If $M$ and $M^*$ are envy-free matchings in a $\TR$ instance and $|M^*| > |M|$, then $M$ admits an envy-free augmenting path. 
\end{lemma}
\begin{proof}
Consider the symmetric difference $M\oplus M^*$.
 There must exist an augmenting path $P=\langle r_1, h_1, r_2, h_2, \dots, r_n, h_n\rangle$ w.r.t. $M$ where
$r_1$ is unmatched and $h_n$ under-subscribed in $M$. Further, for each $i=1,\ldots, n$, $M^*(r_i) = h_i$.
We note that $r_1$ prefers $M^*$ over $M$ (being matched versus being unmatched) and since $M^*$ is envy-free, 
it must be the case that $h_1$ prefers $r_2$ over $r_1$ (else $r_1$
envies $r_2$ w.r.t. $M^*$.) Since $M$ is also envy-free, we conclude that every resident $r_i$ in $P$ prefers $M^*(r_i)$ over $M(r_i)$; 
and every hospital $h_i$ prefers $r_{i+1}$ over $r_i$.

If $M \oplus P$ is envy-free, we are done. Therefore assume
that $M' = M \oplus P$ is not envy-free. Let $r$ have justified envy towards $r'$ w.r.t. $M'$. We first note that if both $r$ and $r'$ belong to $P$,
then $M^*$ is not envy-free. Similarly if both $r$ and $r'$ do not belong to $P$, then $M$ is not envy free. Thus exactly one of $r$ or $r'$ belong to $P$.

We now claim that no resident $r_i$, for $i=1, \ldots, n$ belonging to the path $P$ can have justified envy to any resident outside the path. 
Note that every resident in $P$ gets promoted in $M'$ as compared to $M$. Thus if some $r_i = r$ envies $r'$ w.r.t. $M'$, the same envy pair exists w.r.t. to $M$, a contradiction.
Thus it must be the case that a resident $r$ not belonging to $P$ envies a resident $r' = r_i$ belonging to $P$. We argue that $r$ must be unmatched in $M$.
First note that since $r$ envies $r_i$, 
 there exists an edge $(r, h_i) $ in the graph. Note that
the edge $(r, h_i)$ neither belongs to $M$ nor to $M^*$, because every $h_i$ except $h_n$ is matched in both and $M$ and $M^*$ along residents in $P$. Consider
	$M^*(r) = h$, then there exists an edge $(r, h)$. Since $h_i \neq h$ 
and the length of preference list of $r$ is at most two, we conclude that $r$ must be unmatched in $M$. Thus any resident $r$ that envies a resident $r'$ along $P$ is unmatched in $M$.

We now show that if $M \oplus P$ is not envy-free, we can construct another path $P'$ starting at such an unmatched  resident such that $M \oplus P'$ is envy-free.
Recall the augmenting path $P=\langle r_1, h_1, r_2, h_2, \dots, r_n, h_n\rangle$ w.r.t. $M$. Let $h_i$ be the hospital closest to $h_n$ along this path such
that there exists an unmatched resident $r$ such that $(r, h_i)$ is in the graph and $r$ envies $r_i$ w.r.t. $M'$. If there are multiple such residents, pick the one
that is most preferred by $h_i$. Now consider the path $P' = \langle r, h_i, r_{i+1}, \ldots, r_n, h_n\rangle$. Note that by the choice of the hospital $h_i$,
$M \oplus P'$  is envy-free. This gives us the desired path $P'$.
\end{proof}

\begin{lemma}\label{lem:e1e2_envyfree}
Matching $M$ produced by Algorithm~\ref{algo:maxl-envy-free} is an envy-free matching.
\end{lemma}
\begin{proof}
We first argue that $M$ is envy-free. Assume not.
Consider the first iteration (say $i$-th iteration) during the execution of the algorithm in which an envy-pair is introduced. Call
the matching at the end of iteration $i$ as $M_i$.
Let that envy-pair be $(r, r')$  w.r.t. $M_i$  where $M_i(r') = h$,  and $r >_h r'$ and $h >_r M_i(r)$. 
We consider two cases depending on whether or not the edge $(r', h)$ belongs to $M_{i-1}$.

	\noindent {\em Case 1:} Assume $(r',h) \in M_{i-1}$. In this case, $M_{i-1}(r) >_r h$, 
otherwise $r$ envies $r'$ w.r.t. $M_{i-1}$, a contradiction. 
Now since the algorithm never demotes any resident, it implies that
$M_i(r) >_r h$, contradicts that that $r$ envies $r'$ w.r.t. $M_i$.

	\noindent {\em Case 2:}  Assume $(r',h) \notin M_{i-1}$. If $M_{i-1}(r) <_r h$ 
then either $r$ or a higher preferred resident than $r$ is a threshold for $h$. Thus the edge $(r', h) \in E_1$ and hence gets deleted.
	Else $M_{i-1}(r) >_r h$. 
	In this case by the same argument as above, $r$ does not envy $r'$.

Thus, $M$ is an envy-free matching in $G$.
\end{proof}

\begin{lemma}\label{thm:e1e2_01_rl2}
Algorithm~\ref{algo:maxl-envy-free} produces maximum size envy-free matching if every resident's preference list has length at most $2$ and all the quotas are at most $1$.
\end{lemma}
\begin{proof}
Let $M$ be the output of Algorithm~\ref{algo:maxl-envy-free} on a restricted instance $G$.
Lemma~\ref{lem:e1e2_envyfree} shows that $M$ is envy-free. For the sake of contradiction,
assume that $M$ is not maximum size envy-free. Let $M^*$ be an envy-free matching in $G$ with $|M^*|>|M|$.
Consider the symmetric difference $M\oplus M^*$.
 There must exist an augmenting path $P=\langle r_1, h_1, r_2, h_2, \dots, r_n, h_n\rangle$ w.r.t. $M$ where
$r_1$ and $h_n$ are unmatched in $M$. Moreover $M\oplus P$ is envy-free by Lemma~\ref{lem:structural}. Further, for each $i=1,\ldots, n$, $M^*(r_i) = h_i$.
We note that $r_1$ prefers $M^*$ over $M$ (being matched versus being unmatched) and since $M^*$ is envy-free,
it must be the case that $h_1$ prefers $r_2$ over $r_1$ (else $r_1$
envies $r_2$ w.r.t. $M^*$.) Since $M$ is also envy-free, we conclude that every resident $r_i$ in $P$ prefers $M^*(r_i)$ over $M(r_i)$;
and every hospital $h_i$ prefers $r_{i+1}$ over $r_i$.
If all the $M^*$ edges in this path belong to $E'$ in the final iteration of Algorithm \ref{algo:maxl-envy-free}, then we arrive at a contradiction. Hence there exists some edge $(r_i, h_i) \in  M^* \cap P$ such that $(r_i,h_i)\not\in E'$; that is $(r_i, h_i)$ was deleted
either as an $E_1$ or $E_2$ edge in Step~\ref{step:del} of Algorithm~\ref{algo:maxl-envy-free}.

Suppose $(r_i,h_i)\in E_1$, then there exists a threshold resident for $h_i$, say $r'$ which $h_i$ prefers over $r_i$. Note
that, by definition of threshold resident, $r'$ is matched in $M$ to some hospital $h'$ that it prefers lower over $h_i$.
Since preference lists of residents are at most length two, $r'$ is not adjacent to any other hospital.
We now contradict that $M^*$ is envy-free by observing that
 $q^+(h_i) = 1$ and $M^*(r_i) = h_i$. Thus for $M^*$ to be envy-free, $r'$ must be matched in $M^*$ to some hospital that is higher preferred than $h_i$. However, no such hospital exists,  which implies that  $r'$ must remain unmatched in $M^*$ and thus envies $r_i$.
This implies that $M^*$ is not envy-free; a contradiction.
Now, let $(r_i,h_i)\in E_2$. It implies that $r_i$ was higher preferred by $h_i$ than its threshold resident in the last
iteration and $M(r_i)$ was higher preferred by $r_i$ than $h_i = M^*(r_i)$.
This is a contradiction to the fact that each $r_i$ in $P$ prefers its $M^*(r_i)$ over $M(r_i)$. This completes the proof. 
\end{proof}

%% file: 2-minUR.tex
\subsection{Hardness Results for $\MAXEFM$ and $\MINUR$}\label{sec:hardness}
In order to show the hardness result, we show a reduction from Independent Set ($\IS$) - a well-known $\NP$-complete problem. 
Let $\langle G = (V,E), k\rangle$ be an instance of the
$\IS$ problem where $|V| = n$ and $|E| = m$. The goal in $\IS$
is to decide whether $G$ has an independent set of size  $k$ i.e. a subset of $k$ vertices that are pairwise non-adjacent. 
We create an instance $G' = (\RR \cup \HH, E')$ of the $\MAXEFM$ problem as follows. 
For every vertex $v_i \in V$, we have a vertex-resident $r_i \in \RR$; for every edge $e_j \in E$,
we have an edge-resident $r_j' \in \RR$. Thus $|\RR| = m+n$. 
The set $\HH$ consists of $n+1$ hospitals, one hospital per vertex ($h_i$ for vertex $v_i$) in $G$ and an additional hospital $x$.
The hospital $x$ has both lower-quota and upper-quota as $k$. 
A hospital $h_i$ has zero lower-quota and an upper-quota equal to $1+|E_i|$ where $E_i$ denotes the set of edges incident on $v_i$ in $G$. 
Let $\EE_i$ denote the set of edge-residents corresponding to edges in $E_i$.

\noindent{\bf Preference lists:} The preferences {(which also represent the underlying edge set $E'$)} of the residents and the hospitals can be found in Fig.~\ref{fig:reduction_is}.
A vertex-resident $r_i$ has 
$h_i$ followed by $x$.
An edge-resident $r_j'$ has 
the two hospitals (denoted by $h_{j1}$ and $h_{j2}$) corresponding to the end-points $v_{j1}, v_{j2}$ of the edge $e_j$ in any order.
A hospital $h_i$ has the resident $r_i$ followed by the edge-residents in $\EE_i$ in any strict order.
Finally the hospital $x$ has 
all the $n$ vertex-residents in any strict order.

\noindent {\bf Stable Matching in $G'$:} It is straightforward to verify that a stable matching in $G'$ does not match any resident to $x$, thus making it infeasible. We remark that this property is necessary, otherwise as we prove (see Lemma~\ref{lem:stbl_feas}) that if a stable matching $M_s$ is feasible,
then $M_s$ is itself a $\MAXEFM$.
\begin{figure}[!ht]
	\begin{center}
\begin{minipage}{0.3\columnwidth}
\begin{align*}
    \forall i \in [n],\ r_i &: h_i, x\\
    \forall j \in [m],\ r'_j &: h_{j1}, h_{j2}
\end{align*}
\end{minipage}%
\hfill
\begin{minipage}{0.6\columnwidth}
\begin{align*}
    \forall i \in [n],\ [0,|E_i|+1]\ h_i &: r_i, \EE_i\\
    [k,k]\ x &: r_1, r_2, \ldots, r_n
\end{align*}
\end{minipage}%
	\end{center}
\caption{Preference lists in the reduced instance $G'$ of $\MAXEFM$ from instance $\langle G, k \rangle$ of $\IS$.}
\label{fig:reduction_is}
\end{figure}

\begin{lemma}\label{lem:lem1_is}
$G$ has an independent set of size $k$ iff $G'$ has an envy-free matching of size $m+n$.
\end{lemma}
\begin{proof}
Let $S \subseteq V$ be an independent set of size $k$ in $G$. We construct an envy-free matching
of size $m+n$ in $G'$. 
If $v_i \in S$, match the resident $r_i$ to the hospital $x$. 
When $r_i$ is matched to $x$, any edge-resident $r'_j$ such that edge $e_j$ is incident on $v_i$ cannot be matched to $h_i$, otherwise, $r_i$ envies $r'_j$.
If $v_i \notin S$, match the resident $r_i$ to the hospital $h_i$. 
Since $S$ is an independent set, at least one end-point of every edge is not in $S$. Thus, for an edge $e_j=(v_{j1}, v_{j2})$, 
the corresponding edge-resident $r_j'$ can be matched to at least one of $h_{j1}$ or $h_{j2}$ without causing envy. 
Thus, every vertex-resident and every edge-resident is matched and we have 
an envy-free matching of size $m+n$.

For the other direction, let us assume that $G$ does not have an independent set of size $k$. Consider any envy-free matching $M$ in $G'$. Due to the lower-quota of $x$,
exactly $k$ vertex-residents must be matched to $x$ in $M$. Let $V' = \{v_i \in V \mid M(r_i) = x\}$. Then, $|V'| = k$. Since $V'$ is not an independent set, there exists 
	an edge $e_j = (v_{j1}, v_{j2})$ such that $v_{j1} \in V', v_{j2} \in V'$ that is the residents $r_{j1}$ and $r_{j2}$ are matched to $x$ in $M$. 
This implies that the edge-resident $r_j'$ must be unmatched in $M$, thus $|M| < m+n$. 
\end{proof}

Thus, $\MAXEFM$ is $\NP$-hard. 
This implies that $\MINUR$ is also $\NP$-hard. 
We observe the following for the $\MINUR$ problem. When $G$ has an independent set of size $k$, there are zero residents unmatched in an optimal envy-free matching of $G'$, 
whereas when $G$ does not admit an independent set of
size $k$, every envy-free matching leaves at least one resident unmatched. This immediately implies that there is no $\alpha$-approximation algorithm for $\MINUR$ problem for any $\alpha > 0$. 
Finally, note that in the reduced instance shown in Fig.~\ref{fig:reduction_is}, every resident has exactly two hospitals in its preference list.  
This establishes Theorem~\ref{thm:strong-inapprox}\ref{part1}\ref{parta}.

Above $\NP$-hardness and inapproximability results hold even when the hospital quotas are at most one.
Below we prove this result, however the resident preference lists are no longer of length at most two.
We modify above reduction from $\IS$ problem as follows.
We define the vertex-residents $r_i$, edge-residents $r_j'$ and sets $E_i$, $\EE_i$ as done earlier. Let $q_i = |E_i| +1$ for all vertices $v_i \in V$. For every vertex $v_i \in V$, let $H_i = \{h_i^1, h_i^2, \ldots, h_i^{q_i}\}$ be the set of hospitals corresponding to $v_i$. Let $X = \{x_1, x_2, \ldots, x_k\}$ be also a set of $k$ hospitals. Every hospital in set $H_i$ has zero lower-quota and an upper-quota equal to $1$. Every hospital $x_i \in X$ has both lower and upper-quota equal to $1$.

\noindent{\bf Preference lists:} The preferences of the residents and the hospitals can be found in Fig.~\ref{fig:reduction_is1}. We fix an arbitrary ordering on sets $X$, $H_i$, $\EE_i$. A vertex-resident $r_i$ has the set $H_i$ followed by set $X$. An edge-resident $r_j'$ has two sets of hospitals (denoted by $H_{j1}$ and $H_{j2}$) corresponding to the end-points $v_{j1}, v_{j2}$ of the edge $e_j$ in any order. Every hospital $h \in H_i$ has the vertex -resident $r_i$ as its top-choice followed by the edge-residents in $\EE_i$ in any strict order. Finally the hospitals in $X$ have in their preference lists all the $n$ vertex-residents in any strict order.

\noindent {\bf Stable Matching in $G'$:} It is straightforward to verify that a stable matching in $G'$ does not match any resident to any hospital in set $X$, thus making it infeasible. Recall that this property is necessary.

\begin{figure}[!ht]
\begin{center}
\begin{minipage}{0.3\columnwidth}
\begin{align*}
    \forall i \in [n],\ r_i &: H_i, X\\
    \forall j \in [m],\ r'_j &: H_{j1}, H_{j2}
\end{align*}
\end{minipage}%
\hfill
\begin{minipage}{0.6\columnwidth}
\begin{align*}
	\forall i \in [n],\ t \in [q_i],\ [0,1]\ h_i^t &: r_i, \EE_i\\
	\forall j \in [k],\ [1,1]\ x_j &: r_1, r_2, \ldots, r_n
\end{align*}
\end{minipage}%
	\end{center}
    \caption{Reduced instance $G'$ of $\MAXEFM$ from instance $\langle G, k \rangle$ of $\IS$.}
    \label{fig:reduction_is1}
\end{figure}

Note that every vertex $v_i$ has $q_i$ many hospitals - each with an upper quota of $1$, so the set of hospitals in $H_i$ together have enough quota to get matched with the corresponding vertex-resident $r_i$ and all the edge-residents corresponding to the edges incident on $v_i$. Lemma~\ref{lem:lem_is1} proves the correctness of the reduction. Hence, $\MAXEFM$ and $\MINUR$ are $\NP$-hard even if hospital quotas are at most one.
\begin{lemma}\label{lem:lem_is1}
$G$ has an independent set of size $k$ iff $G'$ has an envy-free matching of size $m+n$.
\end{lemma}
\begin{proof}
	Let $S \subseteq V$ be an independent set of size $k$ in graph $G$. We construct an envy-free matching in $G'$ which matches all the residents in $\RR$. Let $T$ be the set of residents corresponding to the vertices $v_i \in S$ i.e. $T = \{r_i \mid v_i \in S\}$. Match $T$ with $X$ using Gale and Shapley stable matching algorithm~\cite{GS62}. Let $T'$ be the set of residents corresponding to vertices $v_i$ such that $v_i \notin S$ i.e. $T' = \{r_1, r_2, \ldots, r_n\} \setminus T$. Let $H'$ be the set of hospitals appearing in sets $H_i$ such that $v_i \in S$ i.e. $H' = \bigcup\limits_{i: v_i \in S}{H_i}$. Match $T' \cup \{r_1', \ldots, r_m'\}$ with $\HH \setminus H' \setminus X$ using Gale and Shapley stable matching algorithm.

We now prove that the matching is envy-free. No pair of residents in $T$ form an envy-pair because we computed a stable matching between $T$ and $X$. No pair of residents in $T' \cup \{r_1', \ldots, r_m'\}$ form an envy-pair because we computed a stable matching between this set and $H \setminus H' \setminus X$. Since, all hospitals in $H'$ are forced to remain empty, no resident in set $T$ can envy a resident in set $\{r_1', \ldots, r_m'\}$. A resident in $T'$ is matched to a higher preferred hospital than any hospital in $X$, hence such resident cannot envy any resident in $T$. Thus, the matching is envy-free.

We now prove that the matching size is $m+n$. Every vertex-resident $r_i$ is matched either with some hospital in $X$ or some hospital in $H_i$. Since, $S$ is an independent set, at least one end point of every edge is not in $S$. So for every edge $e_t=(v_{t1}, v_{t2})$, there is at least one hospital in sets $H_{t1}$, $H_{t2}$ that can get matched with the edge-resident $r_t'$ without causing envy. Thus, every edge-resident is also matched. Thus, we have an envy-free matching of size $m+n$.

For the other direction, let us assume that $G$ does not have an independent set of size $k$. Consider an arbitrary envy-free matching $M$ in $G'$. Due to the unit lower-quota of every $x_i \in X$, exactly $k$ vertex-residents must be matched to hospitals in $X$. Let $S \subseteq V$ be the set of vertices $v_i$ such that the corresponding vertex-resident $r_i$ is matched to some hospital in $X$ in $M$, i.e. $S = \{v_i \mid M(r_i) \in X\}$. So, $|S| = k$. Since, $S$ is not an independent set, there exists at least two vertex-residents $v_s$ and $v_t$ matched to some hospital in $X$ such that the edge $e_j = (v_s, v_t) \in E$. Due to the preference lists of the hospitals, all the hospitals in both $H_s$ and $H_t$ sets must remain empty in $M$ to ensure envy-freeness. This implies that the edge-resident $r_j'$ must be unmatched. This implies that $|M| < m+n$. This completes the proof of the lemma.
\end{proof}
This establishes Theorem~\ref{thm:strong-inapprox}\ref{part1}\ref{partb}.
From Theorem~\ref{thm:strong-inapprox}\ref{part1} the $\NP$-hardness holds for $\HRLQ$ instances
in which the residents have preference list of length at most $2$ {\em or} hospital quotas are at most one. In the case when {\em both}
the restrictions hold, we show in section~\ref{sec:012R} that $\MAXEFM$ admits a polynomial time algorithm.
Now, we prove our claim that a stable matching, when feasible is a maximum size envy-free matching.
\begin{lemma}\label{lem:stbl_feas}
	A stable matching, when feasible is an optimal solution of $\MAXEFM$.
\end{lemma}
\begin{proof}
	We prove this by showing that an unmatched resident in a stable matching is also unmatched in every envy-free matching.
	Let $M_e$ be an envy-free matching. Since, the set of residents matched in a stable matching is invariant of the matching (by Rural Hospital Theorem~\cite{Roth86}), let's pick an arbitrary stable matching $M_s$. Suppose for the sake of contradiction that resident $r_1$ is matched to hospital $h_1$ in $M_e$ and unmatched in $M_s$. Then, hospital $h_1$ must be full in $M_s$ and $\forall r' \in M_s(h_1), r' >_{h_1} r_1$. 
	In $M_e$ at least one of the residents from $M_s(h_1)$ is not matched to $h_1$. Let that resident be $r_2$. Then envy-freeness of $M_e$ implies that $r_2$ is matched in $M_e$ such that $M_e(r_2) = h_2 >_{r_2} h_1$. 
	By similar argument as earlier, hospital $h_2$ must be full in $M_s$ and $\forall r' \in M_s(h_2), r' >_{h_2} r_2$. 
This process must terminate since there are finite number of residents and each is matched to at most one hospital. But, we prove that such process cannot terminate, implying that the claimed $r_1$ does not exist. 
Since $M_e$ is envy-free, once the process hits a resident $r_i$, it must find a higher preferred hospital $h_i$ than $h_{i-1}$. 
While at a hospital, the process always finds a new resident. 
While at a resident, it may hit some hospital more than once. We prove that in the latter case also, eventually it must find a distinct resident.
	
	Assume that for some resident $r_i$, we have $M_e(r_i) = h_k \in \{h_1, h_2, \ldots, h_{i-2}\}$. 
Hospital $h_k$ is matched to $r_i$ and $r_{k}$ in $M_e$ and matched to $r_{k+1}$ in $M_s$. If $r_{k+1}$ is the only resident matched to hospital $h_k$ in $M_s$, then $(r_k, h_k)$ and $(r_i, h_k)$ block $M_s$. Thus, there must exist another resident $r'$ distinct from $r_1$ to $r_i$ such that $r' >_{h_k} r_i$ and $r' >_{h_k} r_k$ and $r' \in M_s(h_k)$. Thus, we showed that even at the repeated hospital $h_k$, the process must find a distinct resident.
\end{proof}


%% file: 3-inapprox.tex
\subsection{Inapproximability of $\MAXEFM$}
\label{sec:inapprox}

In this section, we show a reduction from the Minimum Vertex Cover ({$\MVC$}) to the $\MAXEFM$ which proves
inapproximability {when hospital quotas are at most one. We note that this result subsumes the $\NP$-hardness result proved in section~\ref{sec:hardness} when hospital quotas are at most one. Nevertheless
the $\NP$-hardness result proved in section~\ref{sec:hardness} additionally hold for the instance when resident list is of length at most two. It also shows strong inapproximability for the $\MINUR$ problem and is also useful in showing W[1]-hardness when deficiency is the parameter (Section~\ref{sec:pc_results}).}

Let $G = (V, E)$ be an instance of {$\MVC$} problem. The goal of {$\MVC$} problem is to find a minimum size vertex cover i.e. a subset $V'$ of vertices such that each edge has at least one end-point included in $V'$.
Our reduction is inspired by the reduction showing inapproximability of the maximum size weakly stable matching 
problem in the presence of ties and incomplete lists ({\sf MAX SMTI}) by Halld\'{o}rsson~et~al.~\cite{HIMY07}. 
The template of our reduction in this section (and also in section~\ref{sec:maxrefm}) is similar to  that in \cite{HIMY07};
however the actual gadgets in both the sections bear no resemblance to the one in \cite{HIMY07}.

\noindent{\bf Reduction:}
Given a graph $G = (V,E)$, which is an instance of the {$\MVC$} problem, we construct an instance $G'$ of the {$\MAXEFM$} problem.
Thus $G'$ is an $\HRLQ$ instance. 
Corresponding to each vertex $v_i$ in $G$, $G'$ contains a gadget with three residents $r_1^i, r_2^i, r_3^i$, and four hospitals $h_1^i, h_2^i, h_3^i, h_4^i$. 
All hospitals have an upper-quota of $1$ and 
$h_3^i$ has a lower-quota of $1$.
Assume that the vertex $v_i$ has $d$ neighbors in $G$, namely  $v_{j_1}, v_{j_2}, \dots, v_{j_d}$. The preference lists of the three residents
and four hospitals, are as in Fig.~\ref{fig:hardness}. We impose an arbitrary but fixed ordering of the neighbors of $v_i$ in $G$ which is used as 
a strict ordering of neighbors in the preference lists of resident $r_3^i$ and hospital $h_1^i$ in $G'$. Note that $G'$ has $N = 3|V|$ residents and $\frac{4N}{3}$ hospitals. 
\begin{figure}[!ht]
\begin{minipage}{0.4\textwidth}
\begin{align*}
    {r_1^i} &: h_1^i\\
    {r_2^i} &: h_2^i, h_1^i, h_3^i\\
    {r_3^i} &: h_4^i, h_2^i, h_1^{j_1}, \dots, h_1^{j_d}, h_3^i
\end{align*}
\end{minipage}%
\hfill
\begin{minipage}{0.5\textwidth}
\begin{align*}
    [0,1]\text{ } {h_1^i} &: r_2^i, r_3^{j_1}, \dots, r_3^{j_d}, r_1^i\\
    [0,1]\text{ } {h_2^i} &: r_2^i, r_3^i\\
    [1,1]\text{ } {h_3^i} &: r_3^i, r_2^i\\
    [0,1]\text{ } {h_4^i} &: r_3^i
\end{align*}
\end{minipage}%
\caption{Preferences of residents and hospitals corresponding to a vertex $v_i$ in $G$ for $\MAXEFM$.}
\label{fig:hardness}
\end{figure}

\begin{lemma}
	The instance $G'$ does not admit any stable and feasible matching.
\end{lemma}
\begin{proof}
The matching $M_s = \{(r_1^i,h_1^i), (r_2^i, h_2^i), (r_3^i, h_4^i)\mid i = 1, \ldots, n\}$ is stable in $G'$ since every resident gets the first choice.
Since $M_s$ leaves $h_3^i$ deficient for each $i$, it is not feasible. By the 
	Rural Hospitals Theorem~\cite{Roth86}
	, we conclude that $G'$ does not admit any stable and feasible matching.
\end{proof}
\begin{lemma}
\label{lem:red-correct}
	Let $G'$ be the instance of the {$\MAXEFM$} problem constructed as above from an instance $G = (V,E)$ of the {$\MVC$} problem. If $VC(G)$ denotes a minimum vertex
	cover of $G$ and $OPT(G')$ denotes a maximum size envy-free matching in $G'$, then $|OPT(G')| = 3|V| - |VC(G)|$.
\end{lemma}
\begin{proof}
We first prove that $|OPT(G')| \ge 3|V| - |VC(G)|$.
Given a minimum vertex cover $VC(G)$ of $G$, we construct an envy-free matching 
	$M$ for $G'$ as follows: $M=\{(r_2^i,h_3^i), (r_3^i,h_4^i)\mid v_i\in VC(G)\}\cup \{(r_1^i,h_1^i), (r_2^i,h_2^i), (r_3^i,h_3^i)\mid v_i\notin VC(G)\}$.
Thus, for a vertex $v_i$ in the vertex cover, $M$ leaves the resident $r_1^i$ unmatched, thereby matching only two residents in the gadget corresponding
to $v_i$. For a vertex $v_i$ that is not in the vertex cover, $M$ matches all the three residents in the gadget corresponding to $v_i$. 
Hence $|OPT(G')| \geq |M|=2|VC(G)|+3(|V|-|VC(G)|)=3|V|-|VC(G)|$. 
\begin{claim}
$M$ is envy-free in $G'$. 
\end{claim}
\begin{claimproof}
It is straightforward to verify that there is no envy-pair consisting of two residents 
associated with the same vertex $v_i \in G$. 
Now, without loss of generality, assume that $r_3^i$ envies a resident matched to hospital $h_1^j$.
By construction of our preference lists, $(v_i, v_j)$ is an edge in $G$. Thus, at least one of $v_i$ or $v_j$ must belong to $VC(G)$.
If $v_i\in VC(G)$, then by the construction of $M$,  $r_3^i$ is matched to its top choice hospital $h_4^i$ in $M$ and hence $r_3^i$
 cannot participate in an envy-pair. Also, $h_1^i$ is left unmatched, hence $r_3^j$ can not form an envy-pair with $M(h_1^i)$.
\end{claimproof}

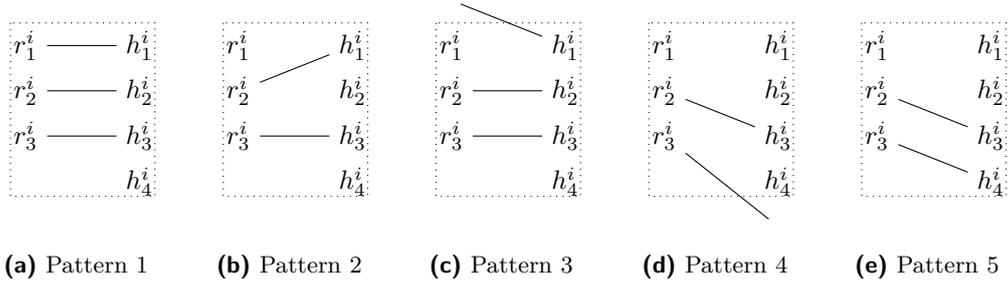
\begin{figure}
  \begin{center}
  \begin{subfigure}{.2\textwidth}
    \begin{tikzpicture}
    \node (A1) at (0,0.6) {};
    \node (A) at (0,0) {$r_1^i$};
    \node (B) at (0,-0.6) {$r_2^i$};
    \node (C) at (0,-1.2) {$r_3^i$};
    \node (a) at (1.5,0) {$h_1^i$};
    \node (b) at (1.5,-0.6) {$h_2^i$};
    \node (c) at (1.5,-1.2) {$h_3^i$};
    \node (x) at (1.5,-1.8) {$h_4^i$};
    \node (c1) at (1.5,-2.4) {};
    \draw[dotted] (-0.2,0.3) rectangle (1.7,-2);
    \path [-] (A) edge node[left] {} (a);
    \path [-] (B) edge node[left] {} (b);
    \path [-] (C) edge node[left] {} (c);
    \end{tikzpicture}
    \caption{Pattern $1$}
    \label{fig:sub2.1}
  \end{subfigure}%
  \begin{subfigure}{.2\textwidth}
    \begin{tikzpicture}
    \node (A1) at (0,0.6) {};
    \node (A) at (0,0) {$r_1^i$};
    \node (B) at (0,-0.6) {$r_2^i$};
    \node (C) at (0,-1.2) {$r_3^i$};
    \node (a) at (1.5,0) {$h_1^i$};
    \node (b) at (1.5,-0.6) {$h_2^i$};
    \node (c) at (1.5,-1.2) {$h_3^i$};
    \node (x) at (1.5,-1.8) {$h_4^i$};
    \node (c1) at (1.5,-2.4) {};
    \draw[dotted] (-0.2,0.3) rectangle (1.7,-2);
    \path [-] (B) edge node[left] {} (a);
    \path [-] (C) edge node[left] {} (c);
    \end{tikzpicture}
    \caption{Pattern $2$}
    \label{fig:sub2.2}
  \end{subfigure}%
  \begin{subfigure}{.2\textwidth}
    \begin{tikzpicture}
    \node (A1) at (0,0.6) {};
    \node (A) at (0,0) {$r_1^i$};
    \node (B) at (0,-0.6) {$r_2^i$};
    \node (C) at (0,-1.2) {$r_3^i$};
    \node (a) at (1.5,0) {$h_1^i$};
    \node (b) at (1.5,-0.6) {$h_2^i$};
    \node (c) at (1.5,-1.2) {$h_3^i$};
    \node (x) at (1.5,-1.8) {$h_4^i$};
    \node (c1) at (1.5,-2.4) {};
    \draw[dotted] (-0.2,0.3) rectangle (1.7,-2);
    \path [-] (A1) edge node[left] {} (a);
    \path [-] (B) edge node[left] {} (b);
    \path [-] (C) edge node[left] {} (c);
    \end{tikzpicture}
    \caption{Pattern $3$}
    \label{fig:sub2.3}
  \end{subfigure}%
  \begin{subfigure}{.2\textwidth}
    \begin{tikzpicture}
    \node (A1) at (0,0.6) {};
    \node (A) at (0,0) {$r_1^i$};
    \node (B) at (0,-0.6) {$r_2^i$};
    \node (C) at (0,-1.2) {$r_3^i$};
    \node (a) at (1.5,0) {$h_1^i$};
    \node (b) at (1.5,-0.6) {$h_2^i$};
    \node (c) at (1.5,-1.2) {$h_3^i$};
    \node (x) at (1.5,-1.8) {$h_4^i$};
    \node (c1) at (1.5,-2.4) {};
    \draw[dotted] (-0.2,0.3) rectangle (1.7,-2);
    \path [-] (B) edge node[left] {} (c);
    \path [-] (C) edge node[left] {} (c1);
    \end{tikzpicture}
    \caption{Pattern $4$}
    \label{fig:sub2.4}
  \end{subfigure}%
  \begin{subfigure}{.2\textwidth}
    \begin{tikzpicture}
    \node (A1) at (0,0.6) {};
    \node (A) at (0,0) {$r_1^i$};
    \node (B) at (0,-0.6) {$r_2^i$};
    \node (C) at (0,-1.2) {$r_3^i$};
    \node (a) at (1.5,0) {$h_1^i$};
    \node (b) at (1.5,-0.6) {$h_2^i$};
    \node (c) at (1.5,-1.2) {$h_3^i$};
    \node (x) at (1.5,-1.8) {$h_4^i$};
    \node (c1) at (1.5,-2.4) {};
    \draw[dotted] (-0.2,0.3) rectangle (1.7,-2);
    \path [-] (B) edge node[left] {} (c);
    \path [-] (C) edge node[left] {} (x);
    \end{tikzpicture}
    \caption{Pattern $5$}
    \label{fig:sub2.5}
  \end{subfigure}%
  \end{center}
  \caption{Five patterns possibly caused by $v_i$}
  \label{fig:reverse}
\end{figure}

Now we prove that $OPT(G')\leq 3|V|-|VC(G)|$.
Let $M=OPT(G')$ be a maximum size envy-free matching in $G'$. 
Consider a vertex $v_i \in V$ and the corresponding residents and hospitals in $G'$. Note that $h_3^i$ must be matched in $M$ for $i=1,\ldots,n$.
Hence following two cases arise. Refer Fig.~\ref{fig:reverse} for the patterns mentioned below.

 {\em Case 1: $M(h_3^i)=r_3^i$.} Then either $M(r_2^i)=h_2^i, M(r_1^i)=h_1^i$ which is pattern $1$  or $M(r_2^i)=h_1^i$ and $r_1^i$ is unmatched (pattern $2$), 
or $M(r_2^i)=h_2^i$ and $M(h_1^i)=r_3^j$ for some $(v_i,v_j)\in E$ (pattern $3$).

 {\em Case 2: $M(h_3^i)=r_2^i$.} Then $(r_1^i,h_1^i)\notin M$, otherwise $r_2^i$ has a justified envy towards $r_1^i$. Also, $(r_3^i,h_2^i)\notin M$ otherwise $r_2^i$ 
has a justified envy towards $r_3^i$. Hence $M(r_3^i)=h_4^i$ (pattern $5$) or $M(r_3^i)=h_1^j$ for some $(v_i,v_j)\in E$ (pattern $4$).

\noindent {\bf Vertex cover $C$ of $G$ corresponding to $M$:}
Using $M$, we now construct the set $C$ of vertices in $G$ which constitute a vertex cover of $G$. If $v_i$ is matched as pattern~1 then $v_i \notin C$, else $v_i \in C$. 
From the following claim, it follows that $C$ is a vertex cover of $G$.
\begin{claim}
If $(v_i,v_j)\in E$, then the gadgets corresponding to both of them can not be matched in pattern $1$ in any envy-free matching $M$.
\end{claim}
\begin{claimproof}
	Let, if possible, there exist an edge $(v_i,v_j)\in E$ such that the gadgets corresponding to both $v_i$ and $v_j$ are matched in pattern~1 in $M$. Thus $M(r_3^i)=h_3^i$ and $M(h_1^j)=r_1^j$. But then $r_3^i$ has justified envy towards $r_1^j$ (via hospital $h_1^j$), contradicting the envy-freeness of $M$.
\end{claimproof}

\noindent {\bf Size of $C$: }
Each gadget could be matched in any of the patterns. Patterns~3 and pattern~4 occur in pairs for a pair of vertices $v_i, v_j$, that is, $M(h_1^i)=r_3^j$ or vice-versa. 
It can be verified that there is no envy-pair among the six residents corresponding to the vertices $v_i, v_j$ matched as pattern~3 and pattern~4 respectively. 
We say that  pattern~3 contributes $2.5$ edges to $M$ and pattern~4 contributes $1.5$ edges. Hence together they contribute to an average matching size of 2. Only pattern $1$ contributes $3$ edges to $M$.
Now it is straightforward to see that  $|OPT(G')|=2|C| + 3(|V|-|C|) = 3 |V| - |C|$. Thus $|VC(G)|\leq |C| = 3|V|-|OPT(G')|$. This completes the proof of the lemma.
\end{proof}

Now we prove the hardness of approximation for the {$\MAXEFM$} problem. 
We assume without loss of generality that an approximation algorithm for the {MAXEFM} problem computes a maximal envy-free matching.
{Lemma~\ref{lem:hardness-approx}} is analogous to Theorem~3.2 and Corollary~3.4 from~\cite{HIMY07}. {Proof of Lemma~\ref{lem:hardness-approx} uses the result of Lemma~\ref{lem:red-correct}}. For the sake of completeness, we give the proof in Appendix~\ref{app:missing_proofs}.
This establishes Theorem~\ref{thm:strong-inapprox}\ref{part2}.

\begin{lemma}
\label{lem:hardness-approx}
It is \NP-hard to approximate the {$\MAXEFM$} problem within a factor of $\frac{21}{19} - \delta$, for any constant $\delta > 0$, even when the quotas of all hospitals are either 0 or 1.
\end{lemma}

%% file: 4-cl.tex
\subsection{Polynomial time algorithm for the $\CL$-restricted instances}\label{sec:cl}
In this section, we consider the $\MAXEFM$ problem on {\sf CL}-restricted {$\HRLQ$} instances with general quotas. 
We first note that every $\HRLQ$ instance with \CL-restriction admits
a feasible envy-free matching. This follows from  the characterization result of Yokoi~\cite{Yokoi20} for instances that admit a feasible envy-free matching.
We now
present a simple modification to the standard Gale and Shapley algorithm~\cite{GS62} that computes a maximum size
envy-free matching. 
Our algorithm (Algorithm~\ref{alg:hrlq_cl}) is based on the $\ESDA$ algorithm presented in~\cite{FITUY15}. In~\cite{FITUY15}, only empirical results without theoretical guarantees on the size of the output matching are presented. Their work also assumes that the underlying graph is complete. We prove that Algorithm~\ref{alg:hrlq_cl} produces maximum size envy-free matching assuming only the $\CL$-restriction. 
We start with an empty matching $M$. Throughout the algorithm, we maintain two parameters:
\begin{itemize}
\item  $d$ :  denotes the deficiency of the matching $M$, that is, 
the sum of deficiencies of all hospitals with positive lower-quota.
\item $k$: the number of unmatched residents w.r.t. $M$.
\end{itemize}
In every iteration, an unmatched resident $r$ who has not yet exhausted its preference list, proposes to the most preferred hospital $h$.
If $h$ is deficient w.r.t. $M$, $h$ accepts $r$'s proposal. 
If $h$ is not deficient, then we consider two cases. Firstly, assume
$h$ is under-subscribed w.r.t. $M$. In this case $h$ accepts the $r$'s proposal only if there are enough unmatched residents to 
satisfy the deficiency of the other hospitals, that is, $k > d$. Next assume that $h$ is fully-subscribed. In this case,
$h$ rejects the least preferred resident in $M(h) \cup r$.
This process continues until some unmatched resident has not exhausted its preference list.
\begin{algorithm}[!ht]
	\SetAlgoNoLine
	\SetAlgoNoEnd
	\KwIn{An $\HRLQ$ instance $G = (\RR \cup \HH, E)$ with $\CL$-restriction}
	\KwOut{Maximum size envy-free matching}
        let $M = \phi$; $\ \ \  d = \sum\limits_{h:q^-(h) > 0}{q^-(h)}$; $\ \ \ k  = |\RR|$\;
	\While{there is an unmatched resident $r$ which has at least one hospital not yet proposed to}{
        $r$ proposes to the most preferred hospital $h$\;
	\If{$|M(h)| < q^-(h)$}{
        	$M = M \cup \{(r, h)\}$\;
		reduce $d$ and $k$ each by $1$\;
	}
	\Else{
		\If{$|M(h)| == q^+(h)$}{
			let $r$' be the least preferred resident in $M(h) \cup r$\;
			$M(h) = M(h) \cup r \setminus r'$\;
		}
		\eIf{$|M(h)| < q^+(h)$ and $k == d$}{ \label{forcereject}
			let $r'$ be the least preferred resident in $M(h) \cup r$\;
			$M(h) = M(h) \cup r \setminus r'$\;
		}
		{
                        \tcp{we have $|M(h)| < q^+(h)$ and $k > d$}
            		$M = M \cup \{(r, h)\}$\;
            		reduce $k$ by $1$\;
		}
	}
	}
	return $M$\;
\caption{$\MAXEFM$ in $\CL$-restricted $\HRLQ$ instances.}
\label{alg:hrlq_cl}
\end{algorithm}

We observe the following about the algorithm. 
Since the input instance is feasible, we start with $k \ge d$ and this inequality is maintained throughout the algorithm. 
If no resident is rejected due to $k = d$ 
in line~\ref{forcereject}, 
then our algorithm degenerates to the Gale and Shapley algorithm~\cite{GS62}
and hence 
outputs a stable matching.
Algorithm~\ref{alg:hrlq_cl} is an adaptation of Gale and Shapley algorithm~\cite{GS62} and runs 
in linear time in the size of the instance.
Lemma~\ref{lem:cl_ef} proves the correctness of our algorithm and establishes Theorem~\ref{thm:CL}.

\begin{lemma}\label{lem:cl_ef}
Matching $M$ computed by Algorithm~\ref{alg:hrlq_cl} is feasible and maximum size envy-free.
\end{lemma}
\begin{proof}
We first prove that the output is feasible. Assume not. Then at termination, $d > 0$, that is, there is at least one
hospital $h$ that is deficient w.r.t. $M$. 
It implies that $k \geq 1$. Thus there is some resident $r$ unmatched w.r.t. $M$. Note that $r$ could not have been rejected by every hospital
	{with positive lower-quota} 
since $h$ appears in the preference list of $r$ and $h$ is deficient at termination.
This contradicts the termination of our algorithm and proves the feasibility of our matching.

Next, we prove that $M$ is envy-free. Suppose for the sake of contradiction, $M$ contains an envy-pair $(r',r)$ 
	such that $(r,h) \in M$ where $r' >_h r$ and $h >_{r'} M(r')$. This implies that
 $r'$ must have proposed to $h$ and $h$ rejected $r'$. If $h$ rejected $r'$ because $|M(h)| = q^+(h)$, $h$ is matched 
with better preferred residents than $r'$, a contradiction to the fact that $r' >_h r$. 
If $h$ rejected $r'$ because $k = d$, then there are two cases. 
Either $r$ was matched to $h$ when $r'$ proposed to $h$.
In this case, in line~\ref{forcereject} our algorithm rejected the least preferred resident in $M(h)$.
This contradicts that $r' >_h r$. Similarly if $r$ proposed to $h$ later, since $k = d$,
the algorithm rejected the least preferred resident again contradicting the presence of any envy-pair.

Finally, we show that $M$ is a maximum size envy-free matching.
We have $k \geq d$ at the start of the algorithm. If during the algorithm, $k=d$ at some point, then at 
the end of the algorithm we have $k=d=0$, implying that, we have an $\RR$-perfect matching and hence the maximum size matching. 
Otherwise, $k > d$ at the end of the algorithm and then we output a stable matching which is maximum size envy-free by Lemma~\ref{lem:stbl_feas}. 
\end{proof}

%% file: 5-param.tex
\section{Envy-freeness: Parameterized complexity}\label{sec:pc_results}
In this section, we investigate parameterized complexity of the $\MAXEFM$ and $\MINUR$ problems. We refer the
reader to the comprehensive literature on parametric algorithms and complexity \cite{DF12,N06,FG06} for standard notation used in this section.
Since the difficulty 
of \MAXEFM\ lies in the instances where stable matchings are not feasible, we choose the parameters related to those hospitals which have a 
positive lower quota (denoted by $H_{LQ}$) 
In particular, the 
deficiency of a given $\HRLQ$ instance (see Definition \ref{def:deficiency} from Section \ref{sec:intro}) is a natural parameter. Unfortunately, the problem turns out to be W[1]-hard for this parameter.
\begin{proof}[Proof of Theorem~\ref{thm:pc_results}\ref{pc_p1}]
	Consider the parameterized version of $\IS$ problem i.e. a graph $G$ and solution size $k$ as the parameter. Let $(G', k')$ be parameterized reduced instance of $\MAXEFM$, where $k'$ is the deficiency of $G'$ and let $k'$ be the parameter. From the NP-hardness reduction given in section~\ref{sec:hardness}, we already saw that the stable matching in $G'$ has deficiency $k$. Then, with $k' = k$, the same reduction is a valid FPT reduction. It implies that that $\MAXEFM$ and $\MINUR$ are W[1]-hard, otherwise it contradicts to the W[1]-hardness of $\IS$~\cite{DF12}. 
\end{proof}

\subsection{A polynomial size kernel}
In this section, we give a kernelization result for $\HRLQ$ instances with hospital quotas either $0$ or $1$.
We consider the following three parameters. 
\begin{itemize}
\item $\ell$: The size of a maximum matching in a given $\HRLQ$ instance.
\item $p$: The highest rank of any \LQ\ hospital in any resident's preference list.
\item $t$: Maximum number of non-\LQ\ hospitals shared by the preference lists of any pair of residents.
\end{itemize}
Given the graph $G$ and $k$ we construct a graph $G'$ such that $G$ admits an envy-free matching of size $k$ iff $G'$ admits an envy-free matching of size $k$.

\noindent {\bf Construction of the graph $G'$:}
We start by computing a stable matching $M_s$ in $G (V, E)$. If $|M_s| < k$, we have a ``No" instance by Lemma~\ref{lem:stbl_feas}. 
If $|M_s| \geq k$ and $M_s$ is feasible, we have a ``Yes" instance. Otherwise, $|M_s| \geq k$ but it is infeasible. 
We know that  $|M_s| \leq \ell$. We construct the graph $G'$ as follows.

Let $X$ be the vertex cover computed by picking matched vertices in  $M_s$. Then, $|X| \leq 2\ell$. Since, $M_s$ is maximal, $I = V \setminus X$ is an independent set.
We now use the marking scheme below to mark edges of $G$ which will belong to the graph $G'$.

\noindent {\bf Marking scheme:} Our marking scheme is inspired by the marking scheme for the kernelization result in \cite{AGRSZ18}. Every edge with both end points in $X$ is marked. 
If $h \in X$ is a hospital, we mark all edges with other end-point in the independent set $I$ if the number of such edges
 are  at most $\ell+1$. Otherwise we mark the edges corresponding to the highest preferred $(\ell+1)$ residents of $h$. 
If $r \in X$ is a resident then we do following: Let $p_r$ denote the highest rank of any \LQ\ hospital in the preference list of $r$.
  Every edge between $r$ and an hospital at rank $1$ to $p_r$ is marked. There can be at most $p$ edges marked in this step. 
We now construct a set of hospitals $C_r$ corresponding to $r$. The set $C_r$ consists of non-$\LQ$ hospitals which are common to
the preference list of $r$ and some matched resident in $M_s$.  That is,
\begin{center}
 $C_r = \{h \in \HH \mid$     $q^-(h) = 0$ and $ \exists r' \in X, r' \neq r$ and $  h$ is in preference list of both $r$ and $r'\}$.
\end{center} Mark all edges of the form $(r, h)$ where $h \in C_r$, if not already marked. Now amongst the unmarked edges incident on $r$ (if any exists)
mark the edge to the highest preferred hospital $h$.
We are now ready to state the reduction rules using the above marking scheme.

\noindent {\bf Reduction rules:} We apply the following reduction rules as long as they are applicable.
\begin{enumerate}
    \item If $v \in G$ is isolated, delete it.
    \item If $(r, h)$ edge is unmarked, delete it. 
\end{enumerate}
Thus, we obtain an instance  $G' = (V', E')$ where $V' = X \cup I$ and $E' = E(X,X) \cup E(X,I)$, where $E(A,B)$ is the set of edges with one end point in $A$ and other in $B$.

Lemma~\ref{lem:3kernel} below bounds the size of the kernel $G'$.
\begin{lemma}\label{lem:3kernel}
The graph $G'$  has $poly(\ell, p, t)$-size.
\end{lemma}
\begin{proof}
We know that $|X| \leq 2 \ell$. 
Thus, $E(X, X) = O(l^2)$. 
Let $X_H \subset X$, be the set of hospitals in $X$ and $X_R \subset X = X \backslash X_H$ be the set of residents in $X$. 
Then, $|X_H| = |X_R| \leq \ell$. 
For a hospital $h \in X_H$, we have at most $\ell+1$ marked edges having its other end-point in the independent set $I$. 
For a resident $r \in X_R$, we retained edges with at most $p+t(\ell-1)+1$ hospitals in independent set $I$. Hence,
 $|E(X,I)| \leq |X_R|\cdot(p+t\ell-t+1) + |X_H|\cdot (\ell+1)= O(\ell(p+t\ell-t+1) + \ell^2)$. 
Since $I$ is independent set, $|I| = |E(X,I)|$. Thus, the size of $G'$ is $O(\ell^2 + \ell(p+t\ell-t+1))$.
\end{proof}

Safeness of first reduction rule is trivial. Lemma~\ref{lem:3lem1} and Lemma~\ref{lem:3lem2} prove that the second reduction rule is safe. So, $G'$ is a kernel.

\begin{lemma}\label{lem:3lem1}
If $G'$ has a feasible envy-free matching $M'$ such that $|M'| \geq k$ then $M'$ is feasible and envy-free in $G$.
\end{lemma}
\begin{proof}
Since, $M' \subseteq E' \subseteq E$, so feasibility in $G$ follows. 
Suppose for the contradiction that $M'$ is not envy-free in $G$. 
Then there exists a deleted edge $(x,y)$ such that it causes envy. 
By the claimed envy, $x$ prefers $y $ over $M'(x)$ and $y$ prefers $x$ over $M'(y)$ .

Suppose $x$ is a hospital. Since, $(x,y)$ is deleted, there are $\ell+1$ marked neighbors of $x$, 
all more preferred than $y$. Since size of maximum matching is at most $\ell$, 
there exists a marked neighbor of $x$, say $y'$ who is unmatched in $M'$. Since, $x$ prefers $y'$ over $y$ it implies, $x$ prefers
$y'$ over $M'(x)$  implying that in $G'$, $y'$ envies $M'(x)$ -- a contradiction since $M'$ is envy-free.

Suppose $x$ is a resident. Given that $(x, y)$ was deleted, $y$ is non-$\LQ$ hospital. Since 
$x$ participates in an envy pair, there are at least two residents $x$ and $M'(y)$ which have a common hospital
$y$ in their preference list. Thus by our marking scheme, $(x, y)$ is not deleted -- a contradiction.
\end{proof}
\begin{lemma}\label{lem:3lem2}
If $G$ has a feasible envy-free matching $M$ such that $|M| \geq k$ then there exists a feasible envy-free matching $M'$ in $G'$ such that $|M'| \geq k$.
\end{lemma}
\begin{proof}
If all the edges in $M$ are present in $G'$ then $M' = M$ and we are done. Suppose not, then there exists an edge $ (x,y) \in M \setminus E'$. 
Let $x$ be a hospital, then since $(x, y)$ was deleted $y \in I$. Note that $y$ is unmatched in $M_s$ in $G$. By  Lemma~\ref{lem:stbl_feas} $y$ cannot be matched
in any envy-free matching, which contradicts that $(x, y) \in M$. Thus $x$ must be a resident.
Since $(x,y)$ is deleted, then there exists a hospital $h$ present only in preference list of $x$ such that $(x,y) \in E'$. 
By the marking scheme, $x$ prefers $h$ over $y$.  Thus, let $M' = M \setminus (x,y) \cup (x, h)$, which is envy-free since there is no other resident in the preference list of 
$h$ other than $x$.
\end{proof}
This establishes Theorem~\ref{thm:pc_results}\ref{pc_p2}.

\subsection{A maximum matching containing a given envy-free matching}
The \MAXEFM\ problem has a polynomial-time algorithm when either there are no \LQ\ hospitals or when all the \LQ\ hospitals have complete preference lists.
This fact suggests two parameters -- number of \LQ\ hospitals in a given instance ($q$), and maximum length of the preference list of 
any \LQ\ hospital $(\ell)$.
Our parameterized algorithm for the parameters $q$ and $\ell$ and other parameterized algorithms, described in Section \ref{sec:fpt}, make crucial use
of an algorithm to {\em extend} an envy-free matching $M$
to a maximum size envy free matching $M^*$, such that $M\subseteq M^*$.
This algorithm was presented in \cite{MNNR18} where it was proved that
it produces a {\em maximal} envy-free matching containing the given envy-free matching $M$. We present the algorithm of \cite{MNNR18} 
for completeness and prove that it outputs a {\em maximum} size envy-free matching containing $M$.

We recall their algorithm below as Algorithm~\ref{algo:sea}. However, unlike Algorithm 2 in~\cite{MNNR18}, where
they start with Yokoi's output~\cite{Yokoi20}, we start with any feasible envy-free matching $M$. 
Since $M$ need not be a minimum size envy-free matching, in line~\ref{line:sea_1}, we set $q^+(h)$ in $G'$ as $q^+(h) - |M(h)|$.
\begin{algorithm}[!ht]
\SetAlgoNoLine
\SetAlgoNoEnd
	\KwIn{Input : $G = (\RR \cup \HH, E)$, $M = $ a feasible envy-free matching in $G$}
Let $\RR'$ be the set of residents unmatched in $M$\\
Let $\HH'$ be the set of hospitals such that $|M(h)|<q^+(h)$ in $G$\\
Let $G'=(\RR'\cup\HH',E')$ be an induced sub-graph of $G$, where $E' = \{(r,h)\mid r\in\RR', h\in\HH',  \textrm{$h$ prefers $r$ over its threshold resident $r_h$} \}$\\
	Set $q^+(h)$ in $G'$ as $q^+(h) - |M(h)|$ in $G$\label{line:sea_1}\\
Each $h$ has the same relative ordering on its neighbors in $G'$ as in $G$\\
	Compute a stable matching $M_s$ in $G'$\label{step:GS}\\
Return $M^*=M\cup M_s$
\caption{Maximum size envy-free matching containing $M$~\cite{MNNR18}}\label{algo:sea}
\end{algorithm}
Lemma~\ref{lem:stbl_feas} proves that a stable matching is a maximum size envy-free matching in an $\HR$ instance.
This is used to prove that the output of the above algorithm is a maximum size envy-free matching containing $M$. 

\begin{lemma}\label{lem:sea_maxm}
The matching $M^*$ output by Algorithm~\ref{algo:sea} is a maximum size envy-free matching with the property $M\subseteq M^*$.
\end{lemma}
\begin{proof}
Assume for the sake of contradiction that $M'$ is an envy-free matching in $G$ such that $M' = M \cup M_x$
and $|M'| > |M^*|$.
We first claim that $M_x \subseteq E'$. If not, then there exists $(r, h) \in M_x$ where $(r, h) \notin E'$. However,
note that $(r, h)$ does not belong to $E'$ implies that there is a threshold resident $r_h$ such that $r_h$ prefers
$h$ over $M(r_h)$ and  $h$ prefers $r_h $ over $r'$.
Thus, $r_h$ has justified envy towards $r$ w.r.t. $M'$ -- this contradicts the assumption that $M'$ is envy-free.

Recall that $G'$ is an $\HR$ instance and $M_s$ is a stable matching in $G'$. To complete the proof it suffices to note that a stable matching in $G'$ is a maximum size envy-free matching in $G'$ by Lemma \ref{lem:stbl_feas}.
\end{proof}

\subsection{$\FPT$ algorithms for $\MAXEFM$}\label{sec:fpt}
In this section, we give $\FPT$ algorithms 
for the \MAXEFM\ problem on several sets of parameters. Our first set of parameters is the number of \LQ\ hospitals $q$ and
the maximum length of the preference list of any \LQ\ hospital $\ell$.
The algorithm is simple: it tries all possible assignments $M_e$ of residents to $\LQ$ hospitals. If some assignment is not envy-free
we discard it. Otherwise we use Algorithm~\ref{algo:sea} to output a maximum size envy-free matching containing $M_e$.
Since our algorithm tries out all possible assignments to $\LQ$ hospitals, and the extension of $M_e$ is a maximum
cardinality envy-free matching containing $M_e$ (by Lemma~\ref{lem:sea_maxm}) it is clear that the algorithm outputs a maximum size
envy-free matching.

	\begin{lemma}
The $\MAXEFM$ problem is \FPT\ when the parameters are the number of \LQ\ hospitals ($q$) and length of the longest preference list of any \LQ\ hospital ($\ell$).
\end{lemma}
\begin{proof}
For an $\LQ$ hospital $h$, there are at most $2^{\ell}$ possible ways of assigning residents to $h$. Since the number of $\LQ$ hospitals
is $q$, our algorithm considers $2^{\ell \cdot q}$ many different matchings. Testing whether a matching $M_e$ is envy-free and to extend it to a maximum size envy-free matching containing $M_e$ using Algorithm~\ref{algo:sea} needs linear time. 
Thus we have an $O^*(2^{\ell \cdot q})$ time algorithm for the $\MAXEFM$ problem. Here $O^*$ hides polynomial terms in $n$ and $m$.
\end{proof}

	We give the following \FPT\ result when the quotas are at most one.
Let $R_d$ be the set of residents that are acceptable to at least one deficient hospital. Let $s = |R_d|$. We denote the deficiency of the given $\HRLQ$ instance by $d$. We prove that the $\MAXEFM$ problem is \FPT\ if parameters are $s$ and $d$.
	\begin{lemma}
The $\MAXEFM$ problem is \FPT\ when the number of deficient hospitals ($d$) and the total number of residents acceptable to deficient hospitals ($s$) are parameters.
	\end{lemma}
\begin{proof}
We use bounded branching algorithm presented in Algorithm~\ref{alg:fpt_sd}. Matching in line~\ref{line:fpt_sd_1} is computed using $\EFHRLQ$ algorithm in~\cite{Yokoi20}.
	For every deficient hospital $h$ (w.r.t. stable matching), we branch on every resident $r$ in the preference list of $h$. Any matching computed along the branch $r$ of $h$ has $(r,h)$. In every branch, we prune the preference lists such that possible envy-pairs w.r.t. current matching are removed. 
If we run out of preference list at a particular level, we mark the branch as ``invalid'' and do not progress on that branch.
This generates a bounded branching tree that has at most $d$ levels and $s$ branches at each level. 
	We process each valid leaf $l$ as follows. Let $A_l$ be the partial matching (assignment) we have computed along the branch that connects $l$ with the root. We compute a stable matching $M$ on the pruned instance $G_l$. Since we removed possible envy-pairs at each level, it is guaranteed that $M_l$ is envy-free. If $M_l$ is not feasible, we discard it otherwise we choose the largest size such matching across all valid leaf nodes as the output.
\begin{algorithm}
\SetAlgoNoLine
\SetAlgoNoEnd
	\KwIn{Input: HRLQ instance containing feasible envy-free matching}
	\KwOut{Output: Maximum size feasible envy-free matching}
	Let $H' = \{h \in H\ |\ h$ is deficient in a stable matching$\}$\\
	$M^* = $ Yokoi's matching\label{line:fpt_sd_1}\\
	\While {$H'$ is not empty}{
        Pick $h \in H'$\\
	\If {Preference list of $h$ is empty}{
		Mark the branch ``invalid''}
        For every resident $r$ in $h$'s preference list, create a branch and match $(r,h)$\\
	In every branch $l$, prune the instance by removing future envy i.e. $E = E \setminus \{(r',h')\ |\ h'$ is more preferred by $r$ than $h$ and $r'$ is less preferred by $h'$ than $r\}$
	}
	\ForEach {valid leaf $l$}{
        Let $A_l$ be the assignment to deficient hospitals\\
	Let $G_l$ be the pruned instance\\
	Compute a stable matching $M$ in $G_l$\\
	Let $M_l = M \cup A_l$\\
	\If {$M_l$ is not feasible}{
		Discard $M_l$\\
		}
	\If {$|M_l| > |M^*|$}{
		$M^* = M_l$\\}
	}
    return $M^*$;
\caption{FPT algorithm for $\MAXEFM$ parameterized in $s, d$}
\label{alg:fpt_sd}
\end{algorithm}
\\\noindent{\bf Correctness:} At every step, the instance is pruned to remove future envy, so it is easy to see that the matching output is envy-free. For the sake of contradiction, assume that there exists another larger envy-free matching $M^*$ than the matching $M$ output by Algorithm~\ref{alg:fpt_sd}. There must exist at least one hospital $h \in H'$ which has at least one different partner in $M$ and $M^*$. But, since we are considering all possible assignments to deficient hospitals, we must have considered the assignment in $M^*$ as well. So, we could not have missed out a larger envy-free matching.\\\\
\noindent{\bf Running time:} 
	It is clear that there are at most $s^d$ possible leaf nodes. Removing future envy and computing stable matching takes $O(m)$. So, overall running time is $O(m\cdot s^d)$.\\\\
Hence, $\MAXEFM$ is FPT if parameters are number of deficient hospitals ($d$) and the total number of unique residents acceptable to deficient hospitals ($s$).
\end{proof}

Now we consider total number of residents acceptable to $\LQ$ hospitals as a parameter and present a parameterized algorithm when quotas are at most one.
	\begin{lemma}
The $\MAXEFM$ problem parameterized on the total number of residents acceptable to \LQ\ hospitals is \FPT.
	\end{lemma}
\begin{proof}
  Consider $R' = \{r \in R\ |\ \exists h \in H_{LQ}\ \text{such that}\ (r,h) \in E\}$. Thus, $R'$ is the set of residents acceptable to at least one $\LQ$ hospital. Algorithm~\ref{alg:fpt_rprime} is FPT for the parameter $|R'|$. Yokoi's matching in line~\ref{line:fpt_rprime_1} is computed using using Yokoi's $\EFHRLQ$ algorithm~\cite{Yokoi20}. Matching in line~\ref{line:fpt_rprime_5} is computed using Algorithm~\ref{algo:sea}.
\begin{algorithm}
	\SetAlgoNoLine
	\SetAlgoNoEnd
	\KwIn{HRLQ instance containing feasible envy-free matching}
	\KwOut{Maximum size feasible envy-free matching}
    $M^* = $ Yokoi's matching \label{line:fpt_rprime_1}\\
	\ForEach {assignment $A$ between $R'$ and $H_{LQ}$}{
		\If {$A$ is not feasible or not envy-free}{
            	discard $A$\\}
	\Else {
            Compute maximum size envy-free matching $M$ containing $A$ \label{line:fpt_rprime_5}\\
	    \If {$|M| > |M^*|$}{
		    $M^* = M$\\}
	    }
	    }
    return $M^*$
\caption{FPT algorithm for $\MAXEFM$ parameterized in $|R'|$}
\label{alg:fpt_rprime}
\end{algorithm}
	\\\noindent{\bf Correctness:} We consider all possible assignments to residents in $R'$ using branching. We discard an assignment that is infeasible or not envy-free. Thus, we consider all possible envy-free and feasible assignments and extend them using Algorithm~\ref{algo:sea}. By Lemma~\ref{lem:sea_maxm}, $M$ is maximum size envy-free matching that contains $A$ and we pick the largest among them.\\\\
\noindent{\bf Running time:} There are $|R'|!$ possible assignments to check. Finding if an assignment is feasible and envy-free takes $O(m)$. Computing maximum size envy-free matching containing a given assignment takes $O(m)$. So, overall running time is $O(m \cdot |R'|!)$.\\\\
Hence, $\MAXEFM$ is FPT if parameter is the number of residents acceptable to $\LQ$ hospitals.
\end{proof}
This establishes Theorem~\ref{thm:pc_results}\ref{pc_p3}.

%% file: 3-maxrsm.tex
\section{Relaxed Stability}\label{sec:maxrsm}
In this section, we present our results related to the relaxed stability in $\HRLQ$ instance.
We prove that $\MAXRSM$ is $\NP$-hard and hard to approximate within a factor of $\frac{21}{19}$ unless $\Poly=\NP$.
Then we present a simple efficient algorithm which gives a 
$\frac{3}{2}$-approximation guarantee for $\MAXRSM$.
\input{10-maxrefm}
\input{11-max-smti}

%% file: 10-maxrefm.tex
\subsection{$\NP$-hardness and inapproximability of $\MAXRSM$}\label{sec:maxrefm}
In this section, we show a reduction from the Minimum Vertex Cover ({$\MVC$}) to the $\MAXRSM$. 
\noindent{\bf Reduction:}
Given a graph $G = (V,E)$, which is an instance of the {\sf MVC} problem, we construct an instance $G'$ of the {$\MAXRSM$} problem.
Let $n = |V|$. 
Corresponding to each vertex $v_i$ in $G$, $G'$ contains a gadget with three residents $r_1^i, r_2^i, r_3^i$, and three hospitals $h_1^i, h_2^i, h_3^i$.
All hospitals have an upper-quota of $1$ and 
$h_3^i$ has a lower-quota of $1$. 
Assume that the vertex $v_i$ has $d$ neighbors in $G$, namely $v_{j_1}, \dots, v_{j_d}$. The preference lists of the three residents
and three hospitals are shown in Fig.~\ref{fig:rsmhardness}. We impose an arbitrary but fixed ordering on the vertices which is used as 
a strict ordering of neighbors in the preference lists of resident $r_1^i$ and hospital $h_2^i$ in $G'$. Note that $G'$ has $N = 3|V|$ residents and hospitals.
\begin{figure}[!ht]
	\centering
\begin{minipage}{0.4\textwidth}
\begin{align*}
    {r_1^i} &: h_3^i, h_2^{j_1}, h_2^{j_2}, \ldots, h_2^{j_d}, h_1^i\\
    {r_2^i} &: h_2^i, h_3^i\\
    {r_3^i} &: h_2^i
\end{align*}
\end{minipage}%
\hfill
\begin{minipage}{0.5\textwidth}
\begin{align*}
    [0,1]\text{ } {h_1^i} &: r_1^i\\
    [0,1]\text{ } {h_2^i} &: r_2^i, r_1^{j_1}, r_2^{j_2}, \ldots, r_2^{j_d}, r_3^i\\
    [1,1]\text{ } {h_3^i} &: r_2^i, r_1^i
\end{align*}
\end{minipage}%
\caption{Preferences of residents and hospitals corresponding to a vertex $v_i$ in $G$.}
\label{fig:rsmhardness}
\end{figure}

\begin{lemma}\label{lem:correctness}
Let $G'$ be the instance of the $\MAXRSM$ problem constructed as above from an instance $G = (V, E)$ of the minimum vertex cover problem. If $V C(G)$ denotes a minimum vertex cover of $G$ and $OPT(G')$ denotes a maximum size relaxed stable matching in $G'$, then $|OPT(G')| = 3|V| - |VC(G)|$.
\end{lemma}
\begin{proof}
We first prove that $|OPT(G')| \geq 3|V| - |VC(G)|$. Given a minimum vertex cover $VC(G)$ of $G$ we construct a relaxed stable matching $M$ for $G'$ as follows. $M = \{(r_1^i, h_3^i), (r_2^i, h_2^i) \mid v_i \in VC(G)\} \cup \{(r_1^i, h_1^i), (r_2^i, h_3^i), (r_3^i, h_2^i) \mid v_i \notin VC(G)\}$. Thus, $|OPT(G')| \geq |M| = 2|VC(G)| + 3(|V| - |VC(G)|) = 3|V| - |VC(G)|$.

\begin{claim}
	{$M$ is relaxed stable in $G'$.}
\end{claim}
\begin{claimproof}
When $v_i \in VC(G)$, residents $r_1^i$ and $r_2^i$ both are matched to their top choice hospitals and hospital $h_2^i$ is matched to its top choice resident $r_2^i$. Thus, when $v_i \in VC(G)$, no resident from the $i$-th gadget participates in a blocking pair. When $v_i \notin VC(G)$, hospitals $h_1^i$ and $h_3^i$ are matched to their top choice residents and we ignore blocking pair $(r_2^i, h_2^i)$ because $r_2^i$ is matched to a lower-quota hospital $h_3^i$, thus there is no blocking pair within the gadget for $v_i \notin VC(G)$. Now suppose that there is a blocking pair $(r_1^i, h_2^j)$ for some $j$ such that $(v_i, v_j) \in E$. Note that either $v_i$ or $v_j$ is in $VC(G)$. If $v_i \in VC(G)$, $r_1^i$ is matched to its top choice hospital $h_3^i$, thus cannot participate in a blocking pair. If $v_i \notin VC(G)$, it implies that $v_j \in VC(G)$. Then for $v_j$'s gadget, $h_2^j$ is matched to its top choice $r_2^j$, thus cannot form a blocking pair.
\end{claimproof}

Now we prove that $OPT(G') \leq 3|V| - |VC(G)|$. Let $M = OPT(G')$ be a maximum size relaxed stable matching in $G'$. Consider a vertex $v_i \in V$ and the corresponding residents and hospitals in $G'$. Refer Fig.~\ref{fig:7P} for the possible patterns caused by $v_i$. Hospital $h_3^i$ must be matched to either resident $r_1^i$ (Pattern $1$) or resident $r_2^i$ (Pattern $2$ to Pattern $7$). If $(r_1^i, h_3^i) \in M$, then the resident $r_2^i$ must be matched to a higher preferred hospital $h_2^i$ in $M$. If $(r_2^i, h_3^i) \in M$ then $h_2^i$ may be matched with either $r_3^i$ or $r_1^j$ of some neighbour $v_j$ or may be left unmatched. Similarly, $r_1^i$ can either be matched to $h_1^i$ or $h_2^j$ of some neighbour $v_j$. This leads to $6$ combinations as shown in Fig.~\ref{fig:7P_2} to Fig.~\ref{fig:7P_7}.
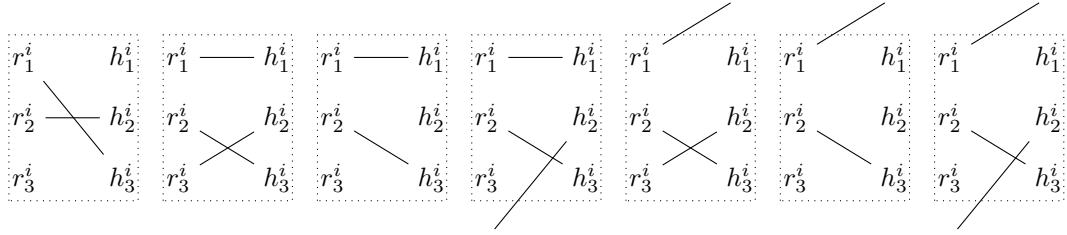
\begin{figure}[!ht]
\begin{center}
    \begin{subfigure}[b]{0.13\textwidth}
    \begin{center}
    \begin{tikzpicture}
    \node (r1) at (0,0.8) {$r_1^i$};
    \node (r2) at (0,0) {$r_2^i$};
    \node (r3) at (0,-0.8) {$r_3^i$};
    \node (h1) at (1.3,0.8) {$h_1^i$};
    \node (h2) at (1.3,0) {$h_2^i$};
    \node (h3) at (1.3,-0.8) {$h_3^i$};
    \path [-] (r1) edge (h3);
    \path [-] (r2) edge (h2);
    \node (x) at (0,-1.6) {};
    \node (y) at (1.3,1.6) {};
    \draw[dotted] (-0.2,-1.1) rectangle (1.5,1.1);
    \end{tikzpicture}
    \caption{Pattern $1$}
    \label{fig:7P_1}
    \end{center}
    \end{subfigure}\hfill
    \begin{subfigure}[b]{0.13\textwidth}
    \begin{center}
    \begin{tikzpicture}
    \node (r1) at (0,0.8) {$r_1^i$};
    \node (r2) at (0,0) {$r_2^i$};
    \node (r3) at (0,-0.8) {$r_3^i$};
    \node (h1) at (1.3,0.8) {$h_1^i$};
    \node (h2) at (1.3,0) {$h_2^i$};
    \node (h3) at (1.3,-0.8) {$h_3^i$};
    \path [-] (r1) edge (h1);
    \path [-] (r2) edge (h3);
    \path [-] (r3) edge (h2);
    \node (x) at (0,-1.6) {};
    \node (y) at (1.3,1.6) {};
    \draw[dotted] (-0.2,-1.1) rectangle (1.5,1.1);
    \end{tikzpicture}
    \caption{Pattern $2$}
    \label{fig:7P_2}
    \end{center}
    \end{subfigure}\hfill
    \begin{subfigure}[b]{0.13\textwidth}
    \begin{center}
    \begin{tikzpicture}
    \node (r1) at (0,0.8) {$r_1^i$};
    \node (r2) at (0,0) {$r_2^i$};
    \node (r3) at (0,-0.8) {$r_3^i$};
    \node (h1) at (1.3,0.8) {$h_1^i$};
    \node (h2) at (1.3,0) {$h_2^i$};
    \node (h3) at (1.3,-0.8) {$h_3^i$};
    \node (x) at (0,-1.6) {};
    \node (y) at (1.3,1.6) {};
    \path [-] (r1) edge (h1);
    \path [-] (r2) edge (h3);
    \draw[dotted] (-0.2,-1.1) rectangle (1.5,1.1);
    \end{tikzpicture}
    \caption{Pattern $3$}
    \label{fig:7P_3}
    \end{center}
    \end{subfigure}\hfill
    \begin{subfigure}[b]{0.13\textwidth}
    \begin{center}
    \begin{tikzpicture}
    \node (r1) at (0,0.8) {$r_1^i$};
    \node (r2) at (0,0) {$r_2^i$};
    \node (r3) at (0,-0.8) {$r_3^i$};
    \node (h1) at (1.3,0.8) {$h_1^i$};
    \node (h2) at (1.3,0) {$h_2^i$};
    \node (h3) at (1.3,-0.8) {$h_3^i$};
    \node (x) at (0,-1.6) {};
    \node (y) at (1.3,1.6) {};
    \path [-] (r1) edge (h1);
    \path [-] (r2) edge (h3);
    \path [-] (x) edge (h2);
    \draw[dotted] (-0.2,-1.1) rectangle (1.5,1.1);
    \end{tikzpicture}
    \caption{Pattern $4$}
    \label{fig:7P_4}
    \end{center}
    \end{subfigure}\hfill
    \begin{subfigure}[b]{0.13\textwidth}
    \begin{center}
    \begin{tikzpicture}
    \node (r1) at (0,0.8) {$r_1^i$};
    \node (r2) at (0,0) {$r_2^i$};
    \node (r3) at (0,-0.8) {$r_3^i$};
    \node (h1) at (1.3,0.8) {$h_1^i$};
    \node (h2) at (1.3,0) {$h_2^i$};
    \node (h3) at (1.3,-0.8) {$h_3^i$};
    \node (x) at (1.3,1.6) {};
    \node (y) at (0,-1.6) {};
    \path [-] (r1) edge (x);
    \path [-] (r2) edge (h3);
    \path [-] (r3) edge (h2);
    \draw[dotted] (-0.2,-1.1) rectangle (1.5,1.1);
    \end{tikzpicture}
    \caption{Pattern $5$}
    \label{fig:7P_5}
    \end{center}
    \end{subfigure}\hfill
    \begin{subfigure}[b]{0.13\textwidth}
    \begin{center}
    \begin{tikzpicture}
    \node (r1) at (0,0.8) {$r_1^i$};
    \node (r2) at (0,0) {$r_2^i$};
    \node (r3) at (0,-0.8) {$r_3^i$};
    \node (h1) at (1.3,0.8) {$h_1^i$};
    \node (h2) at (1.3,0) {$h_2^i$};
    \node (h3) at (1.3,-0.8) {$h_3^i$};
    \node (x) at (1.3,1.6) {};
    \node (y) at (0,-1.6) {};
    \path [-] (r1) edge (x);
    \path [-] (r2) edge (h3);
    \draw[dotted] (-0.2,-1.1) rectangle (1.5,1.1);
    \end{tikzpicture}
    \caption{Pattern $6$}
    \label{fig:7P_6}
    \end{center}
    \end{subfigure}\hfill
    \begin{subfigure}[b]{0.13\textwidth}
    \begin{center}
    \begin{tikzpicture}
    \node (r1) at (0,0.8) {$r_1^i$};
    \node (r2) at (0,0) {$r_2^i$};
    \node (r3) at (0,-0.8) {$r_3^i$};
    \node (h1) at (1.3,0.8) {$h_1^i$};
    \node (h2) at (1.3,0) {$h_2^i$};
    \node (h3) at (1.3,-0.8) {$h_3^i$};
    \node (x) at (0,-1.6) {};
    \node (y) at (1.3,1.6) {};
    \path [-] (r1) edge (y);
    \path [-] (r2) edge (h3);
    \path [-] (x) edge (h2);
    \draw[dotted] (-0.2,-1.1) rectangle (1.5,1.1);
    \end{tikzpicture}
    \caption{Pattern $7$}
    \label{fig:7P_7}
    \end{center}
    \end{subfigure}
    \caption{Seven patterns possibly caused by vertex $v_i$}
    \label{fig:7P}
\end{center}
\end{figure}\raggedbottom
\begin{claim}
	{A vertex cannot cause pattern $5$.}
\end{claim}
\begin{claimproof}
Assume for the sake of contradiction that a vertex $v_i$ causes pattern $5$. Then, there must exist a vertex $v_j$ adjacent to $v_i$ such that $v_j$ causes either pattern $4$ or pattern $7$.

{\em Case 1:} If vertex $v_j$ causes pattern $4$, then $(r_1^j, h_2^i)$ form a blocking pair, a contradiction.

{\em Case 2:} If vertex $v_j$ causes pattern $7$, then there must exist vertices $v_{j+1}, \dots, v_t$ such that there are following edges in $G$ : $(v_i, v_j), (v_j, v_{j+1}), (v_{j+1}, v_{j+2}), \ldots, (v_{t-1}, v_t)$ and vertices $v_j$ to $v_{t-1}$ cause pattern $7$ and $v_t$ causes pattern $4$. See Fig.~\ref{fig:noP5}. In the vertex ordering, we must have $v_{j+1} > v_i$ otherwise $(r_1^j, h_2^i)$ form a blocking pair. But, since $h_2^j$ is matched to $r_1^i$, $v_{j+2} > v_j$. Continuing this way, $v_{t} > v_{t-2}$ but this causes $(r_1^t, h_2^{t-2})$ form a blocking pair. Thus, the claimed set of edges cannot exist.
\begin{figure}[!ht]
    \begin{center}
    \begin{tikzpicture}
    \node (r1i) at (0,0.8) {$r_1^i$};
    \node (r2i) at (0,0) {$r_2^i$};
    \node (r3i) at (0,-0.8) {$r_3^i$};
    \node (h1i) at (1,0.8) {$h_1^i$};
    \node (h2i) at (1,0) {$h_2^i$};
    \node (h3i) at (1,-0.8) {$h_3^i$};
    \draw[dotted] (-0.2,-1) rectangle (1.2,1);
    \node (r1j) at (2,0.8) {$r_1^j$};
    \node (r2j) at (2,0) {$r_2^j$};
    \node (r3j) at (2,-0.8) {$r_3^j$};
    \node (h1j) at (3,0.8) {$h_1^j$};
    \node (h2j) at (3,0) {$h_2^j$};
    \node (h3j) at (3,-0.8) {$h_3^j$};
    \draw[dotted] (1.8,-1) rectangle (3.2,1);
    \node (r1j1) at (4,0.8) {$r_1^{j+1}$};
    \node (r2j1) at (4,0) {$r_2^{j+1}$};
    \node (r3j1) at (4,-0.8) {$r_3^{j+1}$};
    \node (h1j1) at (5,0.8) {$h_1^{j+1}$};
    \node (h2j1) at (5,0) {$h_2^{j+1}$};
    \node (h3j1) at (5,-0.8) {$h_3^{j+1}$};
    \draw[dotted] (3.7,-1) rectangle (5.3,1);
    \node (x) at (7,0) {};
    \node () at (7.5,0) {$\ldots$};
    \path [-] (r1i) edge (h2j);
    \path [-] (r3i) edge (h2i);
    \path [-] (r2i) edge (h3i);
    \path [-] (r1j) edge (h2j1);
    \path [-] (r2j) edge (h3j);
    \path [-] (r2j1) edge (h3j1);
    \path [-] (r1j1) edge (x);
    \node (y) at (8,1) {};
    \node (r1t1) at (10,0.8) {$r_1^{t-1}$};
    \node (r2t1) at (10,0) {$r_2^{t-1}$};
    \node (r3t1) at (10,-0.8) {$r_3^{t-1}$};
    \node (h1t1) at (11,0.8) {$h_1^{t-1}$};
    \node (h2t1) at (11,0) {$h_2^{t-1}$};
    \node (h3t1) at (11,-0.8) {$h_3^{t-1}$};
    \draw[dotted] (9.7,-1) rectangle (11.3,1);
    \node (r1t) at (12,0.8) {$r_1^t$};
    \node (r2t) at (12,0) {$r_2^t$};
    \node (r3t) at (12,-0.8) {$r_3^t$};
    \node (h1t) at (13,0.8) {$h_1^t$};
    \node (h2t) at (13,0) {$h_2^t$};
    \node (h3t) at (13,-0.8) {$h_3^t$};
    \draw[dotted] (11.7,-1) rectangle (13.3,1);
    \path [-] (r1t1) edge (h2t);
    \path [-] (r2t1) edge (h3t1);
    \path [-] (r1j) edge (h2j1);
    \path [-] (r2t) edge (h3t);
    \path [-] (r1t) edge (h1t);
    \path [-] (y) edge (h2t1);
    \end{tikzpicture}
    \caption{Pattern combination that is not relaxed stable if $v_i$ causes pattern $5$.}
    \label{fig:noP5}
\end{center}
\end{figure}
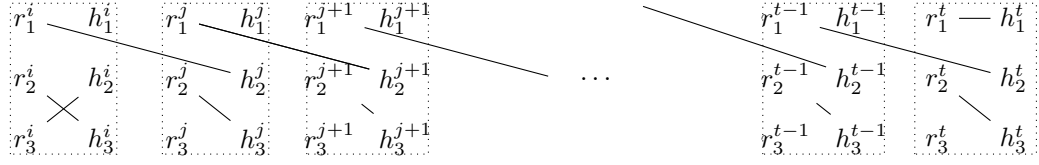
\end{claimproof}
\begin{claim}
	{A vertex cannot cause pattern $3$ or $6$ or $4$.}
\end{claim}
\begin{claimproof}
In pattern $3$ and $6$, $r_3^i$ participates in a blocking pair $(r_3^i, h_2^i)$, contradicting that $M$ is relaxed stable.
If a vertex $v_i$ causes pattern $4$, then there exists a set of $t$ vertices $v_{i+1}, \ldots, v_{i+t}$ such that for $0 \leq k < t, (v_{i+k}, v_{i+k+1})$ is an edge in $G$ and $v_{i+t}$ causes pattern $6$. But, since pattern $6$ cannot occur, pattern $4$ cannot occur.
\end{claimproof}

Thus, a vertex can cause either pattern $1$ or $2$ and thus match all the residents and hospitals within its own gadget or pattern $7$ and match $r_1$ and $h_2$ outside its own gadget. Accordingly there are following cases.

{\em Case 1:} A vertex that causes pattern $7$ can be adjacent to another vertex that causes pattern $7$, which together give matching size $4$ i.e. $2$ per vertex.

{\em Case 2:} It is clear that a vertex causing pattern $1$ or $2$ contributes to matching size of $2$ or $3$ respectively.

\noindent {\bf Vertex cover $C$ of $G$ corresponding to $M$:}
Using $M$, we now construct the set $C$ of vertices in $G$ which constitute a vertex cover of $G$.
If $v_i$ causes pattern $2$, we do not include it in the $C$; Otherwise, we include it. We prove that $C$ is a vertex cover. Suppose not, then there exists an edge $(v_i, v_j)$ such that both $v_i$ and $v_j$ cause pattern $2$. But, this means that $(r_1^i, h_2^j)$ and $(r_1^j, h_2^i)$ form a blocking pair, a contradiction since $M$ is relaxed stable. Now, it is easy to see that $|OPT(G')| = 2|C| + 3(|V| - |C|) = 3|V| - |C|$. Thus, $VC(G) \leq |C| = 3|V| - |OPT(G')|$. This completes the proof of the lemma.
\end{proof}

Now we prove the hardness of approximation for the {$\MAXRSM$} problem. 
Similar to {Lemma~\ref{lem:hardness-approx}}, {Lemma~\ref{lem:hardness-approx_rsm}} is analogous to Theorem 3.2 and Corollary 3.4 from~\cite{HIMY07}. {Proof of Lemma~\ref{lem:hardness-approx_rsm} uses the result of Lemma~\ref{lem:correctness}} and can be reproduced in a similar manner as done in Appendix~\ref{app:missing_proofs} for {Lemma~\ref{lem:hardness-approx}}.
This establishes Theorem~\ref{thm:inapprox_rsm}.
\begin{lemma}
\label{lem:hardness-approx_rsm}
It is $\NP$-hard to approximate the {$\MAXRSM$} problem within a factor of $\frac{21}{19} - \delta$, for any constant $\delta > 0$, even when the quotas of all hospitals are either 0 or 1.
\end{lemma}

%% file: 11-max-smti.tex
\subsection{A $\frac{3}{2}$-approximation algorithm for $\MAXRSM$}
In this section, we present Algorithm~\ref{alg:rsm3by2_1} that computes a relaxed stable matching 
in an $\HRLQ$ instance and prove that it is a $\frac{3}{2}$-approximation  to ${\MAXRSM}$.
Our algorithm
is simple to implement and hence we believe is of practical importance. 
Furthermore, we show that the output of Algorithm~\ref{alg:rsm3by2_1} is at least as large as the stable matching in the instance (disregarding lower-quotas).
Our algorithm is inspired by the  one proposed by Kir{\'{a}}ly~\cite{Kiraly13}. 

We say a feasible matching $M_0$ is {\em minimal} w.r.t. feasibility if for any edge $e \in M_0$, the matching $M_0 \setminus \{e\}$ is
infeasible for the instance. That is for a minimal matching $M_0$, we have for every hospital $h$, $|M_0(h)| = q^-(h)$. Algorithm~\ref{alg:rsm3by2_1} begins by computing a 
feasible matching $M_0$ in the instance $G$ disregarding
the preferences of the residents and hospitals.  Such a feasible matching can be computed by the standard reduction from bipartite matchings to flows with demands on edges ~\cite{KT}.
Let $M = M_0$. We now associate levels with the residents -- all residents matched in $M$ are set to have level-$0$; all residents
unmatched in $M$ are assigned level-$1$. We now execute the Gale and Shapley resident proposing algorithm, with the modification that
a hospital prefers any level-$1$ resident over any level-0 resident (irrespective of the preference list of $h$). Furthermore, if a level-$0$ resident
becomes unmatched during the course of the proposals, then it gets assigned a level-$1$ and it starts proposing from the beginning of its preference
list. Amongst two residents of the same level, the hospital uses its preference list in order them. 
Our algorithm terminates when either every resident is matched or every resident has exhausted its preference list when proposing hospitals at level-1.
It is clear that our algorithm runs in polynomial time since it only computes a feasible matching (using a reduction to flows) and executes a modification of Gale and Shapley algorithm.  
We prove the correctness of our algorithm below.

\begin{algorithm}[!ht]
	\SetAlgoNoLine
	\SetAlgoNoEnd
	\KwIn{Input: $\HRLQ$ instance $G = (\RR \cup \HH, E)$ }
	\KwOut{A relaxed stable matching that is a $\frac{3}{2}$-approximation of $\MAXRSM$}
	$M_0$ is a minimal feasible matching in $G$.
	Let $M = M_0$\;
	For every matched resident $r$, set level of $r$ to level-0;\label{line:star0_0} \\
	For every unmatched resident $r$, set level of $r$ to level-1;\label{line:star0_1} \\
	\While{there is an unmatched resident $r$ which has not exhausted his preference list}{\label{line:while_1}
	$r$ proposes to the most preferred hospital $h$ to whom he has not yet proposed\;
	\If{$h$ is under-subscribed}{
		$M = M \cup \{(r, h)\}$\; \label{line:undersub}
		}
	\Else{
		\If{$M(h)$ has at least one level-0 resident $r'$}{

			$M = M \setminus \{(r',h)\} \cup \{(r,h)\}$\; \label{line:l10}
			Set level of $r'$ to level-1 and $r'$ starts proposing from the beginning of his list; \label{line:star_1}
			}
		\Else{
			$h$ rejects the least preferred resident in $M(h) \cup {r}$\; \label{line:better}
			}
		}
	}
	Return $M$\;
\caption{Algorithm to compute $\frac{3}{2}$-approximation of $\MAXRSM$}
\label{alg:rsm3by2_1}
\end{algorithm}
\begin{lemma}
Matching $M$ output by Algorithm~\ref{alg:rsm3by2_1} is  feasible and relaxed stable.
\end{lemma}
\begin{proof}
We note that $M_0$ is feasible and since Algorithm~\ref{alg:rsm3by2_1} uses a resident proposing algorithm,
it is clear that for any hospital $h$, we have $|M(h)| \geq |M_0(h)| = q^-(h)$. Thus $M$ is feasible.

To show relaxed stability, we claim that when the algorithm terminates, a resident at level-1 does not participate
in a blocking pair. Whenever a level-1 resident $r$ proposes to a hospital $h$, resident $r$ always gets accepted except 
when $h$ is fully-subscribed and all the residents matched to $h$ are level-1 and are better preferred than $r$.
When a  matched level-1 resident $r$ is rejected by a hospital $h$,  $h$ gets a better preferred resident than $r$.
Thus, a level-1 resident does not participate in a blocking pair.
We note that every unmatched resident is a level-1 resident and hence does not participate in a blocking pair. 
Recall that all residents matched in $M_0$ are level-0 residents and $M_0$ is minimal. This implies that 
for every hospital $h$, at most $q^-(h)$ many residents assigned to $h$ in $M_0$ participate in a blocking pair.
We show that in $M$, the number of level-0 residents assigned to any hospital does not increase.
To see this, if $r$ is matched to $h$ in $M$, but not matched to $h$ in $M_0$, it implies that either $r$ was
unmatched in $M_0$ or $r$ was matched to some $h'$ in $M_0$. In either case $r$ becomes level-1 when it gets assigned to $h$ in $M$.
Thus the number of level-0 residents assigned to any hospital $h$ in $M$ is at most $q^-(h)$, all of which can potentially participate in blocking pairs.
This completes the proof that $M$ is relaxed stable.
%
\end{proof}
\begin{lemma}
Matching $M$ output by Algorithm~\ref{alg:rsm3by2_1} is a $\frac{3}{2}$-approximation to the maximum size relaxed stable matching.
\end{lemma}
\begin{proof}
Let $OPT$ denote the maximum size relaxed stable matching in $G$. To prove the lemma we show that in $M \oplus OPT$ there does not exist
any one length as well as any three length augmenting path. To do this, we first convert the matchings $M$ and $OPT$ as one-to-one matchings,
by making {\em clones} of the hospital. In particular we make $q^+(h)$ many copies of the hospital $h$ for every $h$
where the first $q^-(h)$ copies are called \emph{lower-quota copies} and the $q^-(h)+1$ to $q^+(h)$ copies are called  \emph{non lower-quota copies} of $h$.

Let $M_1$ denote the one-to-one matching corresponding to $M$.
To obtain $M_1$, we assign every resident $r \in M(h)$ to a unique copy of $h$ as follows: first, all the residents in $M(h)$ who participate in  blocking pair w.r.t. $M$
 are assigned  unique lower-quota copies of $h$ arbitrarily. The remaining residents in $M(h)$ are assigned to the rest of the copies of $h$, ensuring all lower-quota
copies get assigned some resident.
We get $OPT_1$ from $OPT$ in the same manner.

	Suppose that $(r,h)$ is a one length augmenting path w.r.t. $M$ in  $M \oplus OPT$ such that $r$ is unmatched and $h$ is under-subscribed in $M$.
Recall that an unmatched resident is a level-1 resident, hence $r$ is a level-1 resident. Thus, $r$ must have proposed to $h$ during the execution of algorithm.
Since, $r$ is unmatched, it implies that $h$ must be fully-subscribed in $M$, a contradiction. Thus, there is no one length augmenting path in $M \oplus OPT$.

Next, suppose there exists a three length augmenting path w.r.t. $M$ which starts at an under-subscribed hospital, say $h_j$  and ends at
an unmatched resident in $M$.  Since $h_j$ is under-subscribed in $M$, and there is an augmenting path starting at $h_j$, it implies that
there exists a copy $h_j^d$ such that (i) $h_j^d$ is matched in $OPT_1$ and unmatched in $M_1$, say $OPT_1(h_j^d) = r_d$ and (ii) the resident $r_d$ is matched in $M_1$ (otherwise
there is a one length augmenting path w.r.t. $M_1$, which does not exist); let $M_1(r_d) = h_i^c$, and (iii) the copy $h_i^c$ is matched in $OPT_1$ and $OPT_1(h_i^c) = r_c$ is unmatched in $M_1$ (else
the claimed three length augmenting path does not exist).

We first note that $h_i^c$ and $h_j^c$ are not copies of the same hospital, that is, $i \neq j$, otherwise there is a one length augmenting path $(r_c, h_i)$ w.r.t. $M$.
Since $r_c$ is unmatched in $M_1$ (and hence $M$), the resident $r_c$ is a level-1 resident. Therefore, $r_c$ must have proposed to $h_i$ during the course of the algorithm.
Thus, $h_i$ is fully-subscribed and  is matched to all level-1 residents all of which are better preferred over $r_c$.  This implies that $r_d >_{h_i} r_c$ and  $r_d$ is a level-1 resident.
Since $r_d$ is a level-1 resident, it proposed to hospitals from the beginning of its preference list. Since $h_j$ is under-subscribed, it must be the case that $h_i >_{r_d} h_j$.
Thus, $(r_d, h_i)$ is a blocking pair w.r.t. $OPT$. By the construction of $OPT_1$ from $OPT$, we must have assigned $r_d$ to a lower-quota copy of $h_j$. However, copy $h_j^d$ is
a non lower-quota copy, since it is unassigned in $M_1$, a contradiction. Thus, the claimed three length augmenting path does not exist.
\end{proof}

\begin{wrapfigure}{r}{0.5\textwidth}
    \begin{minipage}{0.22\columnwidth}
    \begin{align*}
	    r_1 &: h_1 \\
	    r_2 &: h_1, h_2 \\
	    r_3 &: h_3, h_2
    \end{align*}
    \end{minipage}%
    \begin{minipage}{0.22\columnwidth}
    \begin{align*}
	    [0,1]\ h_1 &: r_2, r_1 \\
	    [1,1]\ h_2 &: r_2, r_3 \\
	    [0,1]\ h_3 &: r_3
    \end{align*}
    \end{minipage}%
\caption{A tight example for Algorithm~\ref{alg:rsm3by2_1}}
\label{fig:tight}
\end{wrapfigure}
We note that the analysis of our Algorithm~\ref{alg:rsm3by2_1} is tight.
Consider the $\HRLQ$ instance in Fig.~\ref{fig:tight} and a minimal feasible matching $M_0 = \{(r_3,h_2)\}$.
Algorithm~\ref{alg:rsm3by2_1} computes matching $M = \{(r_3,h_2), (r_2,h_1)\}$.
Maximum size relaxed stable matching in this instance is $OPT = \{(r_1,h_1), (r_2,h_2), (r_3,h_3)\}$ and $M \oplus OPT$ admits a five length augmenting path $\langle r_1, h_1, r_2, h_2, r_3, h_3\rangle$.
We also show that every resident matched in stable matching (ignoring lower quotas) is also matched in $M$ that is output by Algorithm~\ref{alg:rsm3by2_1}, implying that $M$ is at least as large as any stable matching.
\begin{lemma}
	A resident matched in a stable matching is also matched in $M$. Hence $M$ is at least as large as any stable matching in that instance.
\end{lemma}
\begin{proof}
Let $M_s$ be a stable matching. 
	By Rural Hospitals theorem, we know that the same set of residents are matched in all the stable matchings.
	Hence, it is enough to prove that a resident $r$ matched in $M_s$ is also matched in $M$.
	Suppose not. Then $r$ must be a level-1 resident. Let $M_s(r) = h$. Since $M$ is relaxed stable, $h$ must be fully-subscribed in $M$ with residents who are level-1 and better preferred over $r$. All these residents in $M(h)$ must be matched and matched to a higher preferred hospital than $h$ in $M_s$ otherwise they form a blocking pair w.r.t. $M_s$. But, since they are level-1, their matched partners in $M_s$ must be fully-subscribed in $M$ with residents who are level-1 and better preferred than them. Thus, a path starting at $r$ who is claimed to be unmatched in $M$ cannot terminate at either a resident or a hospital, a contradiction since there are finite number of hospitals and residents. Hence, every resident matched in $M_s$ is matched in $M$.
\end{proof}

%% file: 12-concl.tex
\section{Discussion}
\label{sec:conl}
In this paper we consider computing matchings with two-sided preferences and lower-quotas. A thorough investigation
of the notion of envy-freeness from a computational perspective reveals that the $\MAXEFM$ problem is $\NP$-hard, and hard to approximate within a constant factor {$\frac{21}{19}$}. In future,
it will be nice to improve the approximation guarantee for the $\MAXEFM$ problem. For the new notion of relaxed stability,
we show desirable properties like guaranteed existence, and an efficient constant factor approximation for the $\MAXRSM$ problem. 
However, the gap between the approximation guarantee and hardness of approximation remains to be bridged.

%% file: 14-appendix.tex
\section{Missing proofs from section~\ref{sec:inapprox} and section~\ref{sec:maxrefm}}\label{app:missing_proofs}
As stated earlier, Lemma~\ref{lem:hardness-approx} is analogous to Theorem 3.2 and Corollary 3.4 from~\cite{HIMY07}.
For completeness, we reproduce the proof below. 
\begin{proposition}~\cite{DS02} \label{prop:dinur}
    For any $\epsilon > 0$ and $p < \frac{3-\sqrt{5}}{2}$, the following statement holds: If there exists a polynomial time algorithm that, given a graph $G = (V,E)$, distinguishes between the following two cases, then $\Poly=\NP$.
    \begin{enumerate}
        \item $|VC(G)| \leq (1-p+\epsilon)|V|$
        \item $|VC(G)| > (1 - max\{p^2, 4p^3-3p^4\}-\epsilon)|V|$
    \end{enumerate}
\end{proposition}

Proposition~\ref{prop:dinur}, Lemma~\ref{lem:red-correct} and $N=3|V|$ together imply following lemma.

\begin{lemma}\label{lem:efm_gap} 
	For any $\epsilon > 0$ and $p < \frac{3-\sqrt{5}}{2}$, the following statement holds: If there exists a polynomial time algorithm that, given a {$\MAXEFM$} instance $G'$ consisting of $N$ residents and $\frac{4N}{3}$ hospitals, distinguishes between the following two cases, then $\Poly=\NP$.
    \begin{enumerate}
        \item $|OPT(G')| \geq \frac{2+p-\epsilon}{3}N$
        \item $|OPT(G')| < \frac{2 + max\{p^2, 4p^3-3p^4\}+\epsilon}{3}N$
    \end{enumerate}
\end{lemma}

\begin{proof}[Proof of Lemma~\ref{lem:hardness-approx}]
As in~\cite{HIMY07}, we substitute $p = \frac{1}{3}$ in Lemma~\ref{lem:efm_gap} to obtain the simplified cases as follows.
\begin{enumerate}
    \item $|OPT(G')|\geq \frac{21-\epsilon}{27} N$
\vspace{0.1in}
    \item $|OPT(G')|< \frac{19+\epsilon}{27} N$    
\end{enumerate}
    Now, suppose we have a polynomial-time approximation algorithm for the {$\MAXEFM$} problem with an approximation factor of at most $\frac{21}{19} - \delta,~\delta > 0$. Consider the above two cases with a fixed constant, $\epsilon < \frac{361\delta}{40-19\delta}$. For an instance of case (1), this algorithm outputs a matching of size $\geq \frac{21-\epsilon}{27}N\frac{1}{\frac{21}{19}-\delta}$, and for an instance of case (2), it outputs a matching of size $<\frac{19+\epsilon}{27} N$. By our setting of $\epsilon$, we can easily verify that $\frac{21-\epsilon}{27}N\frac{1}{\frac{21}{19}-\delta} > \frac{19+\epsilon}{27} N$. Hence using this approximation algorithm, we can distinguish between instances of the two cases, implying that $\Poly=\NP$. This completes the proof of the lemma.
\end{proof}

Following remark is also analogous to Remark 3.6 from~\cite{HIMY07}.
\begin{remark}
A long standing conjecture~\cite{KR08} states that {MVC} is hard to approximate within a factor of $2-\epsilon,~\epsilon>0$. We obtain a lower bound of $1.25$ on the approximation ratio of {$\MAXEFM$}, modulo this conjecture.
\end{remark}
\begin{proof}
Suppose we have an $r$-approximation algorithm for the {$\MAXEFM$} problem. Let $G = (V,E)$ be such that $|VC(G)|\geq \frac{|V|}{2}$. Approximability of {MVC} for general graphs is equivalent to the approximability of {MVC} for graphs with this property~\cite{NT75}.
Using reduction provided 
	obtain a {$\MAXEFM$} instance $G'$. We showed that $|OPT(G')| = 3|V|-|VC(G)|$. Suppose we are given a maximal envy-free matching, $M$ for $G'$, obtained using this algorithm, we can construct a vertex cover $C$ for $G$ with $|C|\leq 3|V|-|M|$. Since $M$ is an $r$-approximation to $OPT(G')$, we have $|M|\geq \frac{|OPT(G')|}{r}$. Combining these constraints we get,
\begin{align*}
    |C| &\leq 3|V| - |M|\\
    &\leq \Big(6-\frac{5}{r}\Big) |VC(G)|
\end{align*}
On substituting $r=1.25-\delta$ in the above equation ($0<\delta \leq 0.25$), we see that $1\leq 6-\frac{5}{r} < 2$. Thus, effectively we have constructed a vertex cover $C$, which is a $k$-approximation to $VC(G)$, where $1\leq k < 2$. This contradicts the conjecture that {MVC} is hard to approximate within a factor of $2-\epsilon,~\epsilon>0$. This completes the proof.
\end{proof}

\noindent Proof of Lemma~\ref{lem:hardness-approx_rsm} can be reproduced in a similar manner as the proof for Lemma~\ref{lem:hardness-approx} {and using the result from Lemma~\ref{lem:correctness} in place of Lemma~\ref{lem:red-correct}}.